\documentclass{lmcs}
\pdfoutput=1
\usepackage[utf8]{inputenc}

\usepackage{lastpage}
\lmcsdoi{20}{3}{23}
\lmcsheading{}{\pageref{LastPage}}{}{}%
{Aug.~31,~2023}{Sep.~10,~2024}{}

\usepackage{quiver}
\usepackage{hyperref}
\usepackage{booktabs}
\usepackage{savesym}
\savesymbol{jmath}

\usepackage{sfmath}
\usepackage{breakurl}             %
\usepackage{underscore}           %
\usepackage{mathpartir}
\usepackage[many]{tcolorbox}
\usepackage{mathtools}
\usepackage{amsmath}
\usepackage{bussproofs}
\savesymbol{amalg}
\usepackage[matha]{mathabx}
\usepackage{amssymb, amsthm}
\usepackage{stmaryrd}
\usepackage{float}
\usepackage{amsfonts}
\savesymbol{mathscr}

\usepackage{mathrsfs}
\usepackage[mathscr]{eucal}
\tcbuselibrary{skins}

\definecolor{darkblue}{RGB}{0,30,160}
\newcommand{\ctb}[1]{\textcolor{darkblue}{\textsf{\textup{#1}}}}
\newcommand{\ctbmath}[1]{\textcolor{darkblue}{#1}}
\newcommand\fv{\textsf{\textup{fv}}}
\newcommand\dom{\textsf{\textup{dom}}}
\newcommand\vrange{\textsf{\textup{vrange}}}

\newcommand{\invred}{\longleftarrow}
\newcommand\redd{\mathrel{{\red}^*}}
\newcommand\iredd{\mathrel{{}^*{\invred}}}
\newcommand{\Dedukti}{\textsc{Dedukti}}
\newcommand{\Type}{\textsf{\textup{Type}}}
\newcommand{\Kind}{\textsf{\textup{Kind}}}
\newcommand{\red}{\longrightarrow}

\newcommand{\trans}[1]{\llbracket #1 \rrbracket}

\newcommand{\UPP}{$\mathbb{UPP}$}
\newcommand{\DkCheck}{\textsc{DkCheck}}
\newcommand{\Agda}{\textsc{Agda}}
\newcommand{\Predicativize}{\textsc{Predicativize}}
\newcommand{\Universo}{\textsc{Universo}}

\newcommand\peq{\overset{\scriptscriptstyle?}{=}}
\newcommand\mgu{m.g.u.\@}
\keywords{Type Theory, Impredicativity, Predicativity, Proof Translation, Universe Polymorphism, Universe-Polymorphic Elaboration, Unification for Universe Levels, Agda, Dedukti}

\begin{document}

\title[Sharing proofs with predicative theories]{Sharing proofs with predicative theories \texorpdfstring{\\}{}through universe-polymorphic elaboration}

\author[T.~Felicissimo]{Thiago Felicissimo\lmcsorcid{0009-0000-1074-9275}}
\author[F.~Blanqui]{Frédéric Blanqui\lmcsorcid{0000-0001-7438-5554}}

\address{Université Paris-Saclay, INRIA project Deducteam, Laboratoire Méthodes Formelles, ENS Paris-Saclay, 91190 France}

\email{thiago.felicissimo@inria.fr, frederic.blanqui@inria.fr}

\begin{abstract}
  As the development of formal proofs is a time-consuming task, it is important to devise ways of sharing the already written proofs to prevent wasting time redoing them. One of the challenges in this domain is to translate proofs written in proof assistants based on impredicative logics to proof assistants based on predicative logics, whenever impredicativity is not used in an essential~way.

  In this paper we present a transformation for sharing proofs with a core predicative system supporting prenex universe polymorphism. It consists in trying to elaborate each term  into a predicative universe-polymorphic term as general as possible. The use of universe polymorphism is justified by the fact that  mapping each universe to a fixed  one in the target theory is not sufficient in most cases. During the elaboration, we need to solve unification problems in the equational theory of universe levels. In order to do this, we give a complete characterization of when a single equation admits a most general unifier. This characterization is then employed in a partial algorithm which uses a constraint-postponement strategy for trying to solve unification problems.

  The proposed translation is of course partial, but in practice allows one to translate many proofs that do not use impredicativity in an essential way.  Indeed, it was implemented in the tool \textsc{Predicativize} and then used to translate semi-automatically many non-trivial developments from \textsc{Matita}'s library to \Agda{}, including proofs of Bertrand's Postulate and Fermat's Little Theorem, which (as far as we know) were not available in \Agda{} yet.
\end{abstract}

\maketitle

\section{Introduction}

An important achievement of the research community in logic is the invention of proof assistants. Such tools allow for interactively writing proofs, which are then checked automatically and  can then be reused in other developments. Proof assistants do not only help mathematicians to make sure that their proofs are indeed correct, but are also used to verify the correctness of safety-critical software.

\paragraph*{Interoperability of proof assistants} Unfortunately, a proof written in a proof assistant cannot be directly reused in another one, which makes each tool isolated in its own library of proofs. This is specially the case when considering two proof assistants with incompatible logics, as in this case simply translating from one syntax to another  would not work. Therefore, in order to share proofs between systems it is often required to do logical transformations.

A naïve approach to share proofs from a proof assistant  $ A $ to a proof assistant $ B $ is to define a transformation acting directly on the syntax of $ A $ and then implement it using the codebase of $ A $. However, this code would be highly dependent on the implementation of $ A $ and can easily become outdated if the codebase of $ A $ evolves. Moreover, if there is another proof assistant $ A' $ whose logic is very similar to the one of $ A $ then this transformation would have to be implemented once again in order to be used with $ A' $ --- the translation is \textit{implementation-dependent}.

\paragraph*{Logical Frameworks \& Dedukti} A better solution is instead to first define the logics of all proof assistants in a common formalism, a \textit{logical framework}. Then, proof transformations can be defined uniformly \textit{inside} the logical framework. Hence, such transformations do not depend on the implementations anymore, but instead on the \textit{logics} that are implemented.

The logical framework \Dedukti{} \cite{dedukti} is a good candidate for a system where multiple logics can be encoded, allowing for logical transformations to be defined uniformly inside \Dedukti{}. Indeed, first, the framework was already shown to be sufficiently expressive to encode the logics of many proof assistants \cite{lmcs:10959}. Moreover, previous works have shown how proofs can be transformed inside \Dedukti{}. For instance, Thiré describes in \cite{sttfa} a transformation to translate a proof of Fermat's Little Theorem from the Calculus of Inductive Constructions to Higher Order Logic (HOL), which can then be exported to multiple proof assistants such as \textsc{HOL}, \textsc{PVS}, \textsc{Lean}, etc. Géran also used \Dedukti{} to export the formalization of Euclid's Elements Book 1 in \textsc{Coq} \cite{geocoq} to several proof assistants \cite{yoan}.

\paragraph*{(Im)Predicativity}

One of the challenges in proof system interoperability is sharing proofs coming from impredicative proof assistants (the majority of them) with predicative ones such as \Agda{}. Indeed, impredicativity, which states that propositions can refer to entities of arbitrary sizes, is a logical principle absent from  predicative systems. It is therefore clear that any proof that uses impredicativity in an essential way cannot be translated to a predicative system. Nevertheless, one can wonder if most proofs written in impredicative systems really use impredicativity and, if not, how one could devise a way for trying to detect this  and translate them to predicative systems.

\paragraph*{A predicativization transformation}

In this paper, we tackle this problem by proposing a transformation that tries to do precisely this. Our translation works by forgetting all the universe information of the initial impredicative term, and then trying to elaborate it into a predicative universe-polymorphic term as general as possible. The need for universe polymorphism arises from the fact that mapping each universe to a unique one in the target theory does not work in most cases --- this is explained in details in Section~\ref{sec:firstlook}.

\paragraph*{Universe level unification}

During the translation, we need to solve level unification problems which are generated when elaborating the impredicative term into a universe-polymorphic one. We therefore develop a (partial) unification algorithm for the equational theory of universe levels. This is done by first giving a novel and complete characterization of which single equations admit a most general unifier (\mgu{}), along with an explicit description of a \mgu{} when it exists. This characterization is then employed in an algorithm implementing a \textit{constraint-postponement} strategy: at each step, we look for an equation admitting a \mgu{} and solve it while applying the obtained substitution to the other equations, in the hope of bringing new ones to the fragment admitting a \mgu{} The given algorithm is \textit{partial} in the sense that, when the unification problem is not a singleton, it may fail to find a \mgu{} even in cases that there is one --- see for instance Example~\ref{exa:mgu}. Our practical results show nevertheless that it is sufficiently powerful for our needs.

\paragraph*{The implementation}

Our predicativization algorithm was  implemented on top of the \DkCheck{} type-checker for \Dedukti{} with the tool \textsc{Predicativize} (available at \url{https://github.com/Deducteam/predicativize}), allowing for the translation of proofs inside \Dedukti{}.  Our tool works in a semi-automatic manner: most of the translation is handled by the proposed algorithm, yet some intermediate steps that are  harder to automate currently require some user intervention. The translated proofs can then be exported to \Agda{}, the main proof assistant based on predicative type theory.

\paragraph*{Translating Matita's arithmetic library}

The tool has been used to translate to the proof assistant \textsc{Agda} the \textit{whole} of \textsc{Matita}'s arithmetic library, making many important mathematical developments available to \textsc{Agda} users. In particular, this work has led to (as far as we know) the first ever proofs in \Agda{} of Fermat's Little Theorem, stating that for $p\in\mathbb{N}$ prime and $n\in\mathbb{N}$ coprime to $p$ we have $n^{p-1}$ equal to $ 1$ modulo $p$, and of Bertrand's Postulate,\footnote{Which, despite its name, is actually a theorem and not a postulate.} stating that for all positive $n\in\mathbb{N}$ one can always find a prime number~$p$ with $n < p \leq 2n$.

The proof of Bertrand's Postulate in \textsc{Matita}  had even been the subject of a whole journal publication~\cite{bertrand}, evidencing its complexity and importance. Thanks to \textsc{Predicativize}, the same hard work did not have to be repeated to make it available in \textsc{Agda}, as the transformation allowed the translation of the whole proof without \textit{any} need of specialist knowledge about it.

\paragraph*{Outline}

We start in Section \ref{sec:dedukti} with an introduction to \Dedukti, before moving to Section \ref{sec:firstlook}, where we present informally the problems that appear when translating proofs to predicative systems. We then introduce in Section \ref{sec:upp} a predicative universe-polymorphic system, which is a subsystem of \Agda{} and is used as the target of the translation.  This is followed by Section \ref{sec:alg}, in which we present the elaboration algorithm. Section \ref{sec:solving} then contributes with a complete characterization of equations admitting a \mgu{}, which is then used to give an algorithm for universe level unification. We then introduce the tool \textsc{Predicativize} in Section \ref{sec:predicativize}, and describe the translation of \textsc{Matita}'s arithmetic library in Section \ref{sec:matita}. Finally,  Section \ref{sec:conc} concludes and discusses  future work.

\paragraph*{Related version}

A preliminary version of this work~\cite{DBLP:conf/csl/FelicissimoBB23} was published in the proceedings of the 31st EACSL Annual Conference
on Computer Science Logic. This journal version contains a number of improvements, among which are the following:

\begin{enumerate}
  \item A main novelty with respect to \cite{DBLP:conf/csl/FelicissimoBB23} is that we provide a complete characterization of when a single equation between universe levels admit a \mgu{} (most general unifier), along with an explicit description of such a \mgu{} This characterization then allows us to give a better algorithm for level unification, which in particular is  complete for singleton problems, whereas the original algorithm is not.

  \item The confluence proof of $\mathbb{UPP}$ in \cite{DBLP:conf/csl/FelicissimoBB23} relied on the \textit{ad hoc} restriction that level variables could only be replaced by levels. We make this condition more precise and integrate it in the definition of \textsc{Dedukti} by adopting a presentation featuring \textit{confinement}, a technique first proposed in \cite{assaf:hal-01515505} which allows to isolate a first-order subset of terms from the higher-order part of the syntax.
  \item Finally, most of the text has been rewritten in order to improve the presentation.
\end{enumerate}

\section{Dedukti}
\label{sec:dedukti}

  \begin{figure}
    \begin{mathpar}
      \inferrule[EmptyCtx]
      {  }
      { \cdot \vdash }
      \and
      \inferrule[ExtCtx]
      { \Gamma \vdash A : \Type }
      { \Gamma, x : A \vdash }
      \and
      \inferrule[ExtCtxC]
      { \Gamma \vdash A : \Type }
      { \Gamma, i : A \vdash }
      \\
      \inferrule[Sort]
      {\Gamma \vdash}
      {\Gamma \vdash \Type : \Kind}
      \and
      x : A \in \Gamma
      \inferrule[Var]
      { \Gamma \vdash }
      {\Gamma \vdash x : A}
      \and
      i : A \in \Gamma
      \inferrule[VarC]
      { \Gamma \vdash }
      {\Gamma \vdash i : A}
      \\
      \begin{matrix}
        &c : A \in \Sigma_\mathbb{T} \text{ or }\hfill\\
        &c :  A :=  u \in \Sigma_\mathbb{T}
      \end{matrix}\hspace{0.4em}
      \inferrule[Cons]
      {\Gamma \vdash}
      {\Gamma \vdash c : A}
      \and
        f : \Delta \to A \in \Sigma_\mathbb{T}
       \inferrule[ConsC]
       {\Gamma \vdash \vec{l}:\Delta}
       {\Gamma \vdash f(\vec{l}) : A[\vec{i}_\Delta \mapsto \vec{l}]}
      \and
      A \equiv B
      \inferrule[Conv]
      {\Gamma \vdash t : A \\ \Gamma \vdash B : s}
      {\Gamma \vdash t : B}
      \and
      A \equiv B
      \inferrule[ConvC]
      {\Gamma \vdash l : A \\ \Gamma \vdash B : \Type}
      {\Gamma \vdash l : B}
      \and
      \inferrule[Arrow]
      {\Gamma\vdash A : \Type\\\Gamma, x : A \vdash B : s}
      {\Gamma \vdash (x : A) \to B : s}
      \and
      \inferrule[ArrowC]
      {\Gamma\vdash A : \Type\\\Gamma, i : A \vdash B : s}
      {\Gamma \vdash (i : A) \to B : s}
      \and
      \inferrule[Abs]
      {\Gamma\vdash A : \Type\\\\\Gamma, x : A \vdash B : s \\ \Gamma, x : A \vdash t : B }
      {\Gamma \vdash x. t : (x : A) \to B}
      \and
      \inferrule[AbsC]
      {\Gamma\vdash A : \Type\\\\\Gamma, i : A \vdash B : s \\ \Gamma, i : A \vdash t : B }
      {\Gamma \vdash i. t : (i : A) \to B}
      \\
      \inferrule[App]
      {\Gamma \vdash t : (x : A) \to B \\ \Gamma \vdash u : A}
      {\Gamma \vdash t~u : B[x \mapsto u]}
      \and
      \inferrule[AppC]
      {\Gamma \vdash t : (i : A) \to B \\ \Gamma \vdash l : A}
      {\Gamma \vdash t~l : B[i \mapsto l]}
    \end{mathpar}
    \caption{Typing rules of \Dedukti}
    \label{fig:typing-dk}
\end{figure}

In this work we use \Dedukti{} \cite{dedukti, lmcs:10959} as the framework in which we express the various type theories we use and define our proof transformation. Therefore, we start with a quick introduction to this system. For a reader familiar with \Dedukti{}, see Remark \ref{rem:dk-comparison} for a comparison of our presentation of \Dedukti{} with more standard ones --- in particular, note that we use a version with \textit{confinement} \cite{assaf:hal-01515505}.

The syntax of  \Dedukti{} is defined by the following grammars. Here, $ c $~ranges over a set of constants $\mathcal{C}$ and $f$ ranges over a set of confined constants $ \mathcal{F} $. Similarly, $ x $~ranges over an infinite set of variables $\mathcal{V}$, whereas $i$ ranges over an infinite set of confined variables $\mathcal{I}$. We assume that the sets $\mathcal{V}$, $\mathcal{I}$, $\mathcal{C}$ and $\mathcal{F}$ are pairwise disjoint, and that each confined constant $f$ comes with an arity $n \in \mathbb{N}$.

Because of confinement, abstraction and application come in two flavors. First, we have $x.t$ and $t~u$ for the usual abstraction and application. Then, we also have $i.t$ and $t~l$ for abstracting a confined variable or applying a regular term to a confined term --- note therefore that confined terms are not terms, but can appear in the right side of an application. Accordingly, we also have the dependent functions types $(x:A) \to B$ for the regular case, and $(i : A) \to B$ for the confined case. Whenever the variable $x$ does not appear free in $B$, we abbreviate $(x:A)\to B$ as just $A \to B$.
\begin{align*}
  &(\textsc{Sorts}) &s,s' &::= \Type \mid\Kind\\
  &(\textsc{Terms and Types}) &A, B, t, u &::= x \mid c \mid s \mid (x : A) \to B \mid x.t \mid t~u \mid (i :A)\to B \mid i.t \mid t~l  \\
  &(\textsc{Confined Terms}) &l, l' &::= i \mid f(l_{1},...,l_{n}) \hspace{3em} \text{where }\textsf{arity}(f)=n
\end{align*}

A \textit{substitution} $\theta$ is a finite set of pairs $x\mapsto t$ or $i \mapsto l$ where each variable occurs at most once --- note that regular variables can only be mapped to regular terms, and confined variables only to confined terms. We write $t[\theta]$ or $l[\theta]$ for its application to a term or confined term, and $\dom(\theta)$ for the set of variables that are assigned to a term or confined term by $\theta$.

A \textit{context} $ \Gamma $ is a finite sequence of entries of the form $ x : A $ or $i:A$ where each variable occurs at most once. A \textit{signature} $ \Sigma $ is a finite sequence of entries of the form $ c : A $ or $ c : A := t$ or $f : (i_{1}:B_{1}..i_{n}:B_{n}) \to A$, where we must have $\textsf{arity}(f)=n$. We adopt the convention of writing the names of constants of the signature in \ctb{blue sans serif font}\footnote{Note that the letters $c$ and $f$ themselves are not written in blue because they are not really names, but rather variables of the metalanguage for referring to constant names.}.

A \textit{rewrite system} $\mathcal{R}$ is a set of \textit{rewrite rules}, which are pairs of the form $c~\vec{t} \red u$ with $\fv(u)\subseteq \fv(c~\vec{t})$. Given a signature $\Sigma$, we also consider the $ \delta $ rules allowing for the unfolding of definitions: we have $ c \red t \in \delta $  for each $ c  :A := t \in \Sigma $. We then denote by $ \red_\mathcal{R} $ the closure by context and substitution of $ \mathcal{R} $, and by $ \red_\delta $ the closure by context of $ \delta $. Finally, we define $\red_{\beta}$ as the closure by context of $(x.t)u \red t[x\mapsto u]$, $\red_{\beta_{c}}$ as the closure by context of $(i.t)l \red t[i \mapsto l]$, and we write $\red_{\beta\beta_{c}\mathcal{R}\delta}$ for $\red_\beta  \cup \red_{\beta_{c}}  \cup \red_\mathcal{R} \cup \red_\delta$.

Rewriting allows us to define equality by computation, but not all equalities  can be defined  like this  in a well-behaved way, e.g. the commutativity of some operator. Therefore, we also consider \textit{rewriting modulo equations} \cite{huet1980confluent,bezem2003term,blanqui03rta}. If $\mathcal{E}$ is a set of equations of the form $l \simeq l'$, we write $\simeq_{\mathcal{E}}$ for its reflexive-symmetric-transitive closure by context and substitution --- note that we only allow equations between confined terms. Because $\mathcal{R}$ and $\mathcal{E}$ are usually kept fixed, in the following we write $\red$ for $\red_{\beta\beta_{c}\mathcal{R}\delta}$, $\simeq$ for $\simeq_{\mathcal{E}}$ and $\equiv$ for the reflexive-symmetric-transitive closure of $\red\cup\simeq$.

One central notion in \Dedukti{} is that of \textit{theory}, which is a triple $  \mathbb{T} = (\Sigma_\mathbb{T}, \mathcal{R}_\mathbb{T},  \mathcal{E}_\mathbb{T}) $ where $ \Sigma_\mathbb{T} $ is a signature and $\mathcal{R}_\mathbb{T}$ and $\mathcal{E}_\mathbb{T}$ are respectively sets of rewrite rules and equations, containing only constants and confined constants declared in $ \Sigma_\mathbb{T} $. Theories are used to define in \Dedukti{} the object logics in which we work (for instance, predicate logic). When working in some theory, we consider untyped terms containing only constants declared in the theory --- that is, we assume that the sets $\mathcal{C}$ and $\mathcal{F}$ contain exactly the constants and confined constants declared in $\Sigma_\mathbb{T}$. Given a theory~$\mathbb{T}$, we define the typing rules of \Dedukti{} as the ones of Figure \ref{fig:typing-dk}. There, we write $\Gamma \vdash \vec{l}:\Delta$ for $\Gamma \vdash l_{k}: B_{k}[i_{1}\mapsto l_{1}..i_{k-1}\mapsto l_{k-1}]$ for all $k=1..n$ when $\Delta = i_{1}:B_{1}..i_{n}:B_{n}$. Note also that the signature  and the conversion relation $\equiv$ are the ones defined by the theory~$\mathbb{T}$.   Finally, whenever the underlying theory is not clear from the context, we write $\Gamma \vdash^{\mathbb{T}} t : A$.

A signature entry $c : A$ or $c : A := t$ or $f : \Delta \to A$ is valid in $\mathbb{T}$ if, respectively, $\vdash^{\mathbb{T}}A : s$ or $\vdash^{\mathbb{T}} t : A$ or $\Delta \vdash^{\mathbb{T}} A : \Type$. A theory $\mathbb{T}$ is then said to be well-formed if each entry in $\Sigma_{\mathbb{T}}$ is valid in $(\Sigma', \mathcal{R}', \mathcal{E}')$, where $\Sigma'$ is the prefix of $\Sigma_{\mathbb{T}}$ preceding the entry in question, and $\mathcal{R}', \mathcal{E}'$ are the restrictions of $\mathcal{R}_{\mathbb{T}}, \mathcal{E}_{\mathbb{T}}$ to rules and equations only containing constants in $\Sigma'$.

\begin{rem}\label{rem:dk-comparison}
  Compared with most presentations of \textsc{Dedukti} \cite{dedukti, lmcs:10959}, ours feature three relevant differences:

  \begin{enumerate}
    \item We consider a syntax with non-annotated abstractions. As shown in \cite{10.5555/645891.671431}, this leads to undecidable type checking. However, in this article we only work with encodings in which the only terms of interest are the $\beta$-normal forms~\cite{felicissimo:LIPIcs.FSCD.2022.25}. For these terms, the omission of such annotations does not jeopardize the decidability of type checking.%
    \item Like some other works \cite{assaf:hal-01515505, guillaume, blanqui22fscd}, we consider a version of \Dedukti{} with rewriting modulo. This is essential to support some equations which cannot be oriented into rewrite rules.
    \item As mentioned previously, we consider a version with \textit{confinement}. This notion, first introduced in \cite{assaf:hal-01515505}, allows to syntactically isolate a first-order part of the syntax from the higher-order one. In \cite{assaf:hal-01515505} it is used to provide a confluence criterion for non-left-linear rewriting. In a similar vein, we use it to allow for non-linear equations in $\mathcal{E}$ without jeopardizing the \textit{Church-Rosser modulo} property --- see Remark~\ref{rem:sim} for further discussion on this point.

  \end{enumerate}
\end{rem}

We recall the following basic metaproperties of \Dedukti{}.

\begin{prop}[Basic metaproperties]\label{prop:basic-meta}~
  Let us write $\Gamma \sqsubseteq \Gamma'$ when $\Gamma$ is a subsequence~of~$\Gamma'$, and let $\Gamma\vdash\mathcal{J}$ range over the typing judgment forms $\Gamma\vdash$ or $\Gamma\vdash t:A$ or $\Gamma\vdash l:A$.
  \begin{description}
  \item[\quad Weakening]  Suppose  $ \Gamma \sqsubseteq \Gamma' $ and $ \Gamma'\vdash$. Then $ \Gamma \vdash \mathcal{J} $ implies $\Gamma' \vdash \mathcal{J}$.
  \item[\quad Substitution property] If $\Gamma,x:B,\Gamma'\vdash \mathcal{J}$ and $\Gamma\vdash u:B$ then $\Gamma,\Gamma'[x\mapsto u]\vdash \mathcal{J}[x\mapsto u]$. If $\Gamma,i:B,\Gamma'\vdash \mathcal{J}$ and $\Gamma\vdash l:B$ then $\Gamma,\Gamma'[i \mapsto l]\vdash \mathcal{J}[i\mapsto l]$.
\end{description}
  In the following points, suppose that the underlying theory is well-formed.
  \begin{description}
  \item[\quad Validity] If $ \Gamma \vdash t : A $ then either $ A = \Kind $ or $ \Gamma \vdash A :s $ for some sort $s$. If $\Gamma\vdash l : A$ then $\Gamma \vdash A : \Type$.
    \item[\quad Subject reduction for $\delta$] If $ \Gamma \vdash t : A $ and $ t \red_\delta t' $ then $ \Gamma \vdash t' : A $
  \end{description}
  We say that injectivity of dependent products holds when $(x : A) \to B \equiv (x : A') \to B'$ implies $A \equiv A'$ and $B \equiv B'$, and $(i : A) \to B \equiv (i : A') \to B'$ implies $A \equiv A'$ and $B \equiv B'$.
  \begin{description}
    \item[\quad Subject reduction for $\beta$ and $\beta_{c}$] If injectivity of dependent products holds, then $ \Gamma \vdash t : A $ and $ t \red_{\beta\beta_{c}} t' $ imply $ \Gamma \vdash t' : A $.
  \end{description}
  A rule $l \red r\in\mathcal{R}$ is said to preserve typing whenever $\Gamma \vdash l[\theta] : A$ implies $\Gamma \vdash r[\theta] :A$, for every $\theta, \Gamma, A$.
  \begin{description}
    \item[\quad Subject reduction for $\mathcal{R}$] If every rule in $\mathcal{R}$ preserves typing, then $\Gamma \vdash t : A$  and $ t \red_{\mathcal{R}} t' $ imply $ \Gamma \vdash t' : A $.
  \end{description}
\end{prop}
\begin{proof}
We refer to \cite{frederic-phd, saillard15phd} for detailed proofs --- even if there the definition of the typing system is not exactly the same, the proofs for the variant used here are straightforward adaptions of their proofs.
\end{proof}

\subsection{Defining type theories in Dedukti}
\label{subsec:pts}

\newcommand\Ty{\ctb{Ty}}
\newcommand\Tm{\ctb{Tm}}
\newcommand\U{\ctb{U}}
\newcommand\nTo{\ctb{$\hspace{0.2em}\leadsto\hspace{0.2em}$}}

We briefly review how one can define type theory with Russell-style universes and dependent products \cite{felicissimo:LIPIcs.FSCD.2022.25}\footnote{Other approaches also exists \cite{dowek2007}, however as argued in \cite{felicissimo:LIPIcs.FSCD.2022.25} they lead to less well-behaved encodings.}. Given a set $\mathfrak{S}$ of \textit{sorts}, we start by declaring the following constants:%
\begin{align*}
&\ctb{Ty}_{s} : \Type &&\text{for $s \in \mathfrak{S}$}\\
&\ctb{Tm}_{s} : \ctb{Ty}_{s} \to \Type &&\text{for $s \in \mathfrak{S}$}
\end{align*}

Then, an object type $A$ at sort $s$ is represented as an element of $\Ty_{s}$, and an object term $t$ of type $A$ is represented as an element of $\Tm_{s}~A$, where $s$ is the sort of $A$. That is, we write $A : \Ty_{s}$ to represent the object-level judgment form $A~\textsf{type}_{s}$, and $t : \Tm_{s}~A$ to represent the object-level judgment form $t :_{s} A$. As an example, if we wanted to add natural numbers at sort $s_{0} \in \mathfrak{S}$, we would add constants $\ctb{Nat} : \Ty_{s_{0}}$, $\ctb{0}: \Tm_{s_{0}}~\ctb{Nat}$, $\ctb{S}: \Tm_{s_{0}}~\ctb{Nat} \to \Tm_{s_{0}}~\ctb{Nat}$, and then a constant for its elimination principle along with its corresponding rewrite rules.

We now add universes. To do this, we suppose we are given a functional relation $\mathfrak{A} \subseteq \mathfrak{S}^{2}$, relating a sort to its successor, and then postulate a constant $\ctb{U}_{s}$ in $\Ty_{s'}$ for each $(s,s') \in \mathfrak{A}$ to represent the universe for the sort $s$. The defining property of the universe for $s$ is that its object terms should correspond somehow to object types at sort $s$. One way of achieving this is by postulating a definitional isomorphism $\Tm_{s'}~\ctb{U}_{s} \simeq \ctb{Ty}_{s}$, which defines \textit{Coquand-style universes}~\cite{coquand-style, kaposi2019gluing, altenkirch2019setoid}. But because in \textsc{Dedukti} we have type-level rewrite rules, we can replace this isomorphism with an identification using a rewrite rule, yielding \textit{Russell-style universes}. In this style of type universes, object terms typed by the universe for $s$ become really the same as the object types at sort $s$.
\begin{align*}
  &\ctb{U}_{s} : \ctb{Ty}_{s'}&&\text{for $(s,s') \in \mathfrak{A}$}\\
  &\ctb{Tm}_{s'}~\ctb{U}_{s} \red \ctb{Ty}_{s}&&\text{for $(s,s') \in \mathfrak{A}$}
\end{align*}

\newcommand\nPi{\ctb{$\Pi$}}
\newcommand\nLam{\ctb{$\lambda$}}
\newcommand\nApp{\ctb{$@$}}

Let us now define dependent products. We now suppose we are given a relation $\mathfrak{R} \subseteq \mathfrak{S}^{3}$, which is functional when seen as  $\mathfrak{R}\subseteq\mathfrak{S}^{2}\times \mathfrak{S}$. Then, for each triple $(s, s',s'')\in\mathfrak{R}$, we postulate a constant $\nPi_{s, s'}$ mapping $A : \Ty_{s}$ and $ B : \Tm_{s}~A \to \Ty_{s'}$ to an element of $\Ty_{s''}$. Then we add abstraction $\nLam_{s,s'}$, mapping $t:(x : \Tm_{s}~A) \to \Tm_{s'}~(B~x)$ to an element of $\Tm_{s''}~(\nPi_{s,s'}~A~B)$, and application $\nApp_{s,s'}$, mapping $t : \Tm_{s''}~(\nPi_{s, s'}~A~B)$ and $u : \Tm_{s}~A$ to an element of $\Tm_{s'}~(B~u)$.
\begin{align*}
  &\nPi_{s,s'} : (A : \Ty_{s}) \to (B : \Tm_{s}~A \to \Ty_{s'}) \to \Ty_{s''}&&\text{for $(s,s',s'') \in \mathfrak{R}$}\\
    &\nLam_{s,s'} : (A : \Ty_{s}) \to (B : \Tm_{s}~A \to \Ty_{s'}) \to \\
    &\hspace{1.5em}((x : \Tm_{s}~A) \to \Tm_{s'}~(B~x)) \to \Tm_{s''}~(\nPi_{s,s'}~A~B)&&\text{for $(s,s',s'') \in \mathfrak{R}$}\\
    &\nApp_{s,s'} : (A : \Ty_{s}) \to (B : \Tm_{s}~A \to \Ty_{s'}) \to \\
    &\hspace{1.5em}(t : \Tm_{s''}~(\nPi_{s,s'}~A~B)) \to (u : \Tm_{s}~A) \to \Tm_{s'}~(B~u)&&\text{for $(s,s',s'') \in \mathfrak{R}$}
\end{align*}

Finally, we need to state the associated computation rule of dependent products. The most natural choice would be to take the following rule.
\begin{align*}
&  \nApp_{s,s'}~A~B~(\nLam_{s,s'}~A~B~t)~u \red t~u &&\text{for $(s,s',s'') \in \mathfrak{R}$}
\end{align*}However, this rule is non left-linear, and thus the rewrite system it generates is non-confluent on raw terms \cite{klop2022combinatory}. Instead, the standard solution is to \textit{linearize} \cite{blanqui05mscs, saillard15phd, mellies1996generic} the rule, by replacing the second occurrences of variables $A, B$ by new fresh variables $A', B'$.
\begin{align*}
  &  \nApp_{s,s'}~A~B~(\nLam_{s,s'}~A'~B'~t)~u \red t~u &&\text{for $(s,s',s'') \in \mathfrak{R}$}
\end{align*}
Note that the left-hand side is not well-typed anymore, but this is not a problem. Indeed, the important property we can show is that for all well-typed instances of the left-hand side, $A$ becomes convertible to $A'$ and $B$ becomes convertible to $B'$, which then guarantees that the corresponding instance of the right-hand side is well-typed with the same type as the left one.

This concludes the definition of the theory, which is summarized in Figure \ref{fig:dk-pts}. In the following, we adopt some conventions to make the notation lighter. First, we write $\nPi_{s, s'}  x : A. B$ for $\nPi_{s, s'}~A~(x. B)$, or $A \hspace{0.2em}\ctbmath{\leadsto}_{s,s'}\hspace{0.1em} B$ when $x \not \in \fv(B)$. Moreover, in order to keep examples readable it will be essential to treat some arguments  as implicit: namely, we will write $\Tm~A$ instead of $\Tm_{s}~A$, $\nPi x : A.B$ instead of $\nPi_{s,s'} x : A.B$, $t_{\nApp} u$ instead of $\nApp_{s,s'}~A~B~t~u$ and $\nLam x.t$ instead of $\nLam_{s,s'}~A~B~(x.t)$. Nevertheless, note that this is just an informal notation used in examples, and in actual \Dedukti{} terms all these arguments should be spelled out --- alternatively, one can switch to a framework with support for \textit{erased arguments}, such~as~\cite{felicissimo2023generic}.

\begin{figure}
  \begin{align*}
    &\ctb{Ty}_{s} : \Type &&\text{for $s \in \mathfrak{S}$}\\
    &\ctb{Tm}_{s} : \ctb{Ty}_{s} \to \Type &&\text{for $s \in \mathfrak{S}$}\\
    &\ctb{U}_{s} : \ctb{Ty}_{s'}&&\text{for $(s,s') \in \mathfrak{A}$}\\
    &\ctb{Tm}_{s'}~\ctb{U}_{s} \red \ctb{Ty}_{s}&&\text{for $(s,s') \in \mathfrak{A}$}\\
    &\nPi_{s,s'} : (A : \Ty_{s}) \to (B : \Tm_{s}~A \to \Ty_{s'}) \to \Ty_{s''}&&\text{for $(s,s',s'') \in \mathfrak{R}$}\\
    &\nLam_{s,s'} : (A : \Ty_{s}) \to (B : \Tm_{s}~A \to \Ty_{s'}) \to \\
    &\hspace{1.5em}((x : \Tm_{s}~A) \to \Tm_{s'}~(B~x)) \to \Tm_{s''}~(\nPi_{s,s'}~A~B)&&\text{for $(s,s',s'') \in \mathfrak{R}$}\\
    &\nApp_{s,s'} : (A : \Ty_{s}) \to (B : \Tm_{s}~A \to \Ty_{s'}) \to \\
    &\hspace{1.5em}(t : \Tm_{s''}~(\nPi_{s,s'}~A~B)) \to (u : \Tm_{s}~A) \to \Tm_{s'}~(B~u)&&\text{for $(s,s',s'') \in \mathfrak{R}$}\\
    &  \nApp_{s,s'}~A~B~(\nLam_{s,s'}~A'~B'~t)~u \red t~u &&\text{for $(s,s',s'') \in \mathfrak{R}$}
  \end{align*}
  \caption{\Dedukti{} theory defined by specification $(\mathfrak{S}, \mathfrak{A}, \mathfrak{R})$}
  \label{fig:dk-pts}
\end{figure}

\begin{rem}\label{rem:dk-pts} In the Pure Type Systems (PTS) \cite{pts} literature, the triple $(\mathfrak{S}, \mathfrak{A}, \mathfrak{R})$ is known as a (functional) \textit{PTS specification}, and each such specification defines a PTS. The main result of \cite{felicissimo:LIPIcs.FSCD.2022.25} is that the theory in Figure \ref{fig:dk-pts} defines an adequate encoding of the associated PTS in \Dedukti{}.
\end{rem}

\section{An informal look at predicativization}
\label{sec:firstlook}

In this informal section we present the problem of proof predicativization and discuss the challenges that arise through the use of examples. Even though the examples might be simplistic, they showcase real problems we found during our first predicativization attempt of Fermat's little theorem library in HOL --- some of them being already noted in~\cite{delort:hal-02985530}.

We first start by defining the impredicative theory $\mathbb{I}$ and the predicative theory $\mathbb{P}$, which will serve respectively as source and target of our proof transformation. They are defined by instantiating the theory of Figure \ref{fig:dk-pts} with the following specifications --- following Remark~\ref{rem:dk-pts}, they can equivalently be seen as the Pure Type Systems defined by these specifications.
\begin{align*}
  &\mathfrak{S}_{\mathbb{I}} := \{ \Omega, \square\} & &\mathfrak{S}_{\mathbb{P}} := \mathbb{N}\\
  &\mathfrak{A}_{\mathbb{I}} := \{ (\Omega,  \square)\} & & \mathfrak{A}_{\mathbb{P}} := \{(n, n+1) \mid n \in \mathbb{N}\}\\
  &\mathfrak{R}_{\mathbb{I}} := \{ (\Omega,\Omega,\Omega),(\square,\Omega,\Omega), (\square, \square, \square)\} & & \mathfrak{R}_{\mathbb{P}} := \{(n,m, \textsf{max}\{n,m\}) \mid n, m \in \mathbb{N}\}
\end{align*}

Note that in the theory $\mathbb{I}$ we have $(\square,\Omega,\Omega)\in\mathfrak{R}_{\mathbb{I}}$, and thus for $\Gamma \vdash A : \Ty_{\square}$ and $\Gamma, x : \Tm~A \vdash B : \Ty_{\Omega}$ we have $\Gamma \vdash \nPi x :A.B : \Ty_{\Omega}$. Therefore, the sort $\Omega$ is closed under dependent products indexed over types in $\square$, a larger sort, so $\mathbb{I}$ is indeed an impredicative theory. Finally, we note that $\mathbb{P}$ is a subtheory of the one implemented in \textsc{Agda}, whereas $\mathbb{I}$ is a subtheory of the ones implemented in \textsc{Coq}, \textsc{Matita}, \textsc{Isabelle}, etc, justifying why they are of interest.

In the following, let $\Phi$ be a signature without confined constant declarations. We say that $\Phi$ is well-formed in a theory $\mathbb{T}$ when the theory $(\Sigma',\mathcal{R}_{\mathbb{T}},\mathcal{E}_{\mathbb{T}})$ is well-formed, where $\Sigma':=\Sigma_{\mathbb{T}},\Phi$. We then call $\Phi$ a \textit{local} signature, in contrast to $\Sigma_{\mathbb{T}}$ which is \textit{global}. We  write names of constants in the local signature in \textsf{sans serif black} in order to distinguish them from names in the global signature, which are still written in \ctb{sans serif blue}.

Then, the problem of proof predicativization consists in defining a transformation such that, given a local signature $ \Phi $ well-formed in $\mathbb{I}$, we obtain a local signature $ \Phi' $ well-formed in $\mathbb{P}$  --- a suitable transformation should of course preserve the structure of the statements in $\Phi$, however we leave the precise relationship between $\Phi$ and $\Phi'$ vague at this point. Stated informally, we would like to translate constants declarations $c : A$ (which represent axioms) and constant definitions $c:A:=t$ (which also represent proofs) from $\mathbb{I}$ to $\mathbb{P}$ --- confined constant declarations $f: \Delta\to A$ cannot occur in $\Phi$ and will be of no interest here. Note in particular that such a transformation is not applied to a single term but to a sequence of constants and definitions, which can be related by dependency. This dependency, as we shall see, turns out to be a major issue for the translation.

Now that our basic notions are explained, let us  try to predicativize proofs. For our first step,  consider a very simple development showing that for every object term $ A $ in $ \ctb{U}_{\Omega} $ we can build an object term in $ A \nTo A $ --- this is actually just the polymorphic identity function at sort  $ \Omega $.  Here we adopt an \textsc{Agda}-like syntax to display entries of the local signature.
\begin{flalign*}
  &\textsf{id} : \Tm~(\nPi A : \U_{\Omega}. A \nTo A)\\
  &\textsf{id} := \nLam A. \nLam x. x
\end{flalign*}

To translate this simple development, the first idea that comes to mind is to define a mapping on sorts: the sort $ \Omega $ is mapped to $ 0 $ and the universe $ \square $ is mapped to $ 1 $. If this mapping defined a \textit{specification morphism}\footnote{That is, if $(s, s') \in \mathfrak{A}_{\mathbb{I}}$ implied $(\phi(s),\phi(s'))  \in \mathfrak{A}_{\mathbb
  P}$ and $(s,s',s'') \in \mathfrak{R}_{\mathbb{I}}$ implied $(\phi(s),\phi(s'),\phi(s'')) \in \mathfrak{R}_{\mathbb{P}}$, where $\phi : \mathfrak{S}_{\mathbb{I}} \to \mathfrak{S}_{\mathbb{P}}$ is the sort mapping.} then this transformation would always produce a valid definition in $\mathbb{P}$~\cite[Lemma 4.2.6]{geuvers1993logics}. Unfortunately, it is easy to check that it does not define a specification morphism (worse, no function $\mathfrak{S}_{\mathbb{I}} \to \mathfrak{S}_{\mathbb{P}}$ defines a specification morphism). Nevertheless, this does not mean that it cannot produce something well-typed in $\mathbb{P}$ in \textit{some} cases. For instance, by applying it to \textsf{id} we get the following entry,\footnote{Modulo the recomputation of some omitted sort annotations.} which is actually well-typed in $\mathbb{P}$.
\begin{flalign*}
  &\textsf{id} : \Tm~(\nPi A : \U_{0}. A \nTo A)\\
  &\textsf{id} := \nLam A. \nLam x. x
\end{flalign*}

This naïve approach however quickly fails when considering other cases. For instance, suppose now that one adds the following definition.
\begin{align*}
  &\textsf{id-to-id} : \Tm~((\nPi A : \U_{\Omega}. A \nTo A) \nTo (\nPi A : \U_{\Omega}. A \nTo A))\\
  &\textsf{id-to-id} := \textsf{id}_{\nApp}(\nPi A : \U_{\Omega}. A \nTo A)
\end{align*}

If we try to perform the same syntactic translation as before, we get the following result.
\begin{align*}
  &\textsf{id-to-id} : \Tm~((\nPi A : \U_{0}. A \nTo A) \nTo (\nPi A : \U_{0}. A \nTo A))\\
  &\textsf{id-to-id} := \textsf{id}_{\nApp}(\nPi A : \U_{0}. A \nTo A)
\end{align*}

However, one can verify that this term is not well typed. Indeed, in the original term one quantifies over all elements of $ \U_{\Omega} $ in  $\nPi A : \U_{\Omega}. A \nTo A $, and because of impredicativity this object term stays at $ \U_{\Omega} $. However, in $\mathbb{P}$ quantifying over all elements of the universe $ \U_{0} $ in $ \nPi A : \U_{0}. A \nTo A $ lifts its overall object type to $ \U_{1} $. As $ \textsf{id} $ expects something of object type $ \U_{0} $, then $  \textsf{id}_{\nApp}(\nPi A : \U_{0}. A \nTo A)$ is not well-typed.

The takeaway lesson from this first try is that impredicativity introduces a kind of \textit{typical ambiguity}, as it allows us to put in a single universe $ \U_{\Omega} $ types which, in a predicative setting, would have to be stratified and placed in larger universes. Therefore, we should not translate every occurrence of $ \Omega $ as $ 0 $ naively as we did, but try to compute for each occurrence of $ \Omega $ some natural number $ n $ such that replacing it by $ n $ would produce a valid term in $\mathbb{P}$. In other words, we should erase all sort information and then \textit{elaborate} it  into a well-typed term in $\mathbb{P}$.

Thankfully, performing such kind of transformations is exactly the goal of the tool \textsc{Universo}~\cite{thire}. To understand how it works, let us come back to the previous example. \Universo{} starts here by replacing each sort by a fresh metavariable representing a natural number.
\begin{align*}
  &\textsf{id} : \Tm~(\nPi A : \U_{i_{1}}. A \nTo A)\\
  &\textsf{id} := \nLam A. \nLam x. x\\
  &\textsf{id-to-id} : \Tm~((\nPi A : \U_{i_{2}}. A \nTo A) \nTo (\nPi A : \U_{i_{3}}. A \nTo A))\\
  &\textsf{id-to-id} := \textsf{id}_{\nApp}(\nPi A : \U_{i_{4}}. A \nTo A)
\end{align*}

Then, in the following step \Universo{} tries to elaborate the term into a well-typed one in $\mathbb{P}$. To do so, it first tries to  typecheck it and generates constraints in the process. These constraints are then given to a SMT solver, which is used to compute for each metavariable $ i $ a natural number so that the local signature is valid in $\mathbb{P}$. For instance, applying \Universo{} to our previous example would produce the following local signature, which is indeed valid with respect to $\mathbb{P}$.
\begin{align*}
  &\textsf{id} : \Tm~(\nPi A : \U_{1}. A \nTo A)\\
  &\textsf{id} := \nLam A. \nLam x. x\\
  &\textsf{id-to-id} : \Tm~((\nPi A : \U_{0}. A \nTo A) \nTo (\nPi A : \U_{0}. A \nTo A))\\
  &\textsf{id-to-id} := \textsf{id}_{\nApp}(\nPi A : \U_{0}. A \nTo A)
\end{align*}

By using \textsc{Universo} it is possible to go much further than with the naïve method shown before. Still, this approach also fails when being employed with real libraries. To see the reason, consider the following minimum example, in which one uses $\textsf{id}$ twice to build another element of the same type.
\begin{align*}
  &\textsf{id}' : \Tm~(\nPi A : \U_{\Omega}. A \nTo A)\\
  &\textsf{id}' := \textsf{id}_{\nApp}(\nPi A : \U_{\Omega}. A \nTo A)_{\nApp} \textsf{id}
\end{align*}

If we repeat the same procedure as before, we get the following entries, which when type checked generate unsolvable constraints.
\begin{align*}
  &\textsf{id} : \Tm~(\nPi A : \U_{i_{1}}. A \nTo A)\\
  &\textsf{id} := \nLam A. \nLam x. x\\
  &\textsf{id}' : \Tm~(\nPi A : \U_{i_{2}}. A \nTo A)\\
  &\textsf{id}' := \textsf{id}_{\nApp}(\nPi A : \U_{i_{3}}. A \nTo A)_{\nApp} \textsf{id}
\end{align*}

The reason is that the application $ \textsf{id}_{\nApp}(\nPi A : \U_{i_{3}}. A \nTo A)_{\nApp} \textsf{id}$  forces $ i_{1} $ to be both $ i_{3} $ and $ i_{3} +1  $, which is of course impossible. Therefore, the takeaway lesson from this second try is that impredicativity does not only hide the fact that types need to be stratified, but also the fact that they need to be usable at multiple levels of this stratification. Indeed, in our example we would like to use $\textsf{id}$ at $\nPi A : \U_{i_{3}}. A \nTo A$ and  $\nPi A : \U_{i_{3}+1}. A \nTo A$. In practice, when trying to translate libraries using \textsc{Universo} we found that at very early stages a translated proof or object was already needed at multiple universes at the same time, causing the translation to fail.

Therefore, in order to properly compensate for the lack of impredicativity, we should not translate entries by fixing once and for all their universes, but instead we should let them vary by using \textit{universe polymorphism}~\cite{typechecking-with-universes,coq}. This feature, present in some type theories (and also in the one of \textsc{Agda}~ \cite{agda}), allows defining terms containing universe variables, which can later be instantiated at various concrete universes.

Our translation will then work by first computing for each definition or declaration its set of constraints. However, instead of assigning concrete values to metavariables, we perform \textit{unification} which allows us to solve constraints in a symbolic way. The result will then be a universe-polymorphic term, which will be usable at multiple universes when translating the next entries. In order to define this formally, we first start in the next section by refining the target type theory $\mathbb{P}$ with universe polymorphism.

\section{A universe-polymorphic predicative theory}
\label{sec:upp}

\newcommand\nLvl{\ctb{Lvl}}
\newcommand\nZero{\ctb{$0$}}
\newcommand\nSucc{\ctb{S}}
\newcommand\nMax{\ctb{$\hspace{0.2em}\sqcup\hspace{0.2em}$}}

In this section we define $\mathbb{UPP}$, a theory which refines $\mathbb{P}$ by internalizing sort annotations and allowing for prenex universe polymorphism \cite{typechecking-with-universes,guillaume}. This is in particular a subsystem of the one underlying the \textsc{Agda} proof assistant \cite{agda}.

The main change with respect to $\mathbb{P}$ is that, instead of indexing constants externally, we index them inside the framework \cite{assaf,sterling2019algebraic}. To do this, we first introduce in \Dedukti{} a syntax for  \textit{universe levels} (which is just the terminology used in the literature for predicative sorts, i.e. $\mathbb{N}$), using the following declarations. Note that the constants $\nZero, \nSucc, \nMax$ are declared as confined, which will be needed later to show the confluence of the theory.
\begin{align*}
  &\nLvl : \Type &&\nSucc : (i:\nLvl) \to \nLvl\hspace{1em}\text{(confined)}\\
  &\nZero : () \to \nLvl\hspace{1em}\text{(confined)} &&\nMax : (i:\nLvl, i':\nLvl) \to \nLvl\hspace{1em}\text{(confined)}
\end{align*}

In the following, we write $\nMax$ in infix notation, $\nSucc$ in curryfied notation, and consider $\nMax$ as having a lower precedence than  $\nSucc$  --- for instance, $\nSucc~i\nMax \nSucc~j$ should be parsed as $\nMax (\nSucc(i), (\nSucc(j)))$. These constant declarations yield the following grammar of confined terms, which from now on we refer to as \textit{levels}.\[
  l, l' ::= i \mid \nZero \mid \nSucc~l \mid l \nMax l'
\]

The definitions of Figure \ref{fig:dk-pts} are then replaced by the following ones. Note that the two rewrite rules are presented in linearized form, in order for them to be left-linear.
\begin{align*}
  &\ctb{Ty} : (l : \nLvl) \to\Type\\
    &\ctb{Tm} : (l : \nLvl)  \to \ctb{Ty}~l \to \Type\\
    &\ctb{U} : (l : \nLvl) \to \ctb{Ty}~(\nSucc~l)\\
    &\ctb{Tm}~l'~(\ctb{U}~l) \red \Ty~l\\
    &\nPi : (l~l' : \nLvl) \to (A : \Ty~l) \to (B : \Tm~l~A \to \Ty~l') \to \Ty~(l \nMax l')\\
    &\nLam : (l~l' : \nLvl) \to(A : \Ty~l) \to (B : \Tm~l~A \to \Ty~l') \to \\
    &\hspace{1.5em}((x : \Tm~l~A) \to \Tm~l'~(B~x)) \to \Tm~(l \nMax l')~(\nPi~l~l'~A~B)\\
    &\nApp :  (l~l' : \nLvl) \to(A : \Ty~l) \to (B : \Tm~l~A \to \Ty~l') \to \\
    &\hspace{1.5em}(t : \Tm~(l \nMax l')~(\nPi~l~l'~A~B)) \to (u : \Tm~l~A) \to \Tm~l'~(B~u)\\
    &  \nApp~l~l'~A~B~(\nLam~l''~l'''~A'~B'~t)~u \red t~u
\end{align*}

In the following, we adopt a subscript notation for levels and write $\Ty_{l}$, $\Tm_{l}$, $\ctb{U}_{l}$, $\nPi_{l, l'}$, $\nLam_{l,l'}$ and $\nApp_{l, l'}$ to improve clarity. We also continue to write $\nPi_{l, l'} x : A. B$ for $\nPi_{l,l'}~A~(x.B)$ or $A \hspace{0.2em}\ctbmath{\leadsto}_{l,l'}\hspace{0.1em} B$ when $x \not \in \fv(B)$. Finally, in order to keep examples readable we also reuse our convention for implicit arguments and write $\Tm~A$ for $\Tm_{l}~A$,  $\nPi x : A . B$ for $\nPi_{l,l'} x :A.B$,  $\nLam x.t$ for $\nLam_{l,l'}~A~B~(x.t)$, and $t_{\nApp} u$ for $\nApp_{l, l'}~A~B~t~u$.

Universe polymorphism can now be represented directly with the use of the framework's function type \cite{assaf}. Indeed, if a definition contains free level variables, it can be made universe-polymorphic by abstracting over such variables. The following example illustrates~this.

\begin{exa}
  The universe-polymorphic identity function  is given  by
  \begin{align*}
    &\ctb{id} : (i : \nLvl) \to \Tm~(\nPi A : \ctb{U}_{i}. A \nTo A)\\
    &\ctb{id} := i. \nLam A. \nLam a. a
  \end{align*}
  This then allows to use $ \textsf{id} $ at any universe level: for instance, we can obtain the polymorphic identity function at the level $ \nZero$ with the application $ \textsf{id}~\nZero$, which has type $ \Tm~(\nPi A : \ctb{U}_{\nZero}. A \nTo A)$.
\end{exa}

\begin{rem}
Note that unlike some other proposals \cite{bezem2023type} there is no object-level operation for universe level abstraction, which is instead handled by the framework function type. Therefore, universe-polymorphic definitions are best understood as schemes rather than actual terms in the object logic.%

\end{rem}

In order to finish the definition of the theory we need to specify the definitional equality satisfied by levels, which is the one generated by the following equations~\cite{agda}.\footnote{Some authors consider a version of this theory without the neutral element $\nZero$~\cite{bezem2023type}; here we stick to the variant  used in \Agda{}, which includes the $\nZero$. } This then concludes the definition of the theory $\mathbb{UPP}$.
\begin{align*}
  i_{1}\nMax (i_{2} \nMax i_{3}) &\simeq (i_{1}\nMax i_{2}) \nMax i_{3} &\nSucc~(i_{1} \nMax i_{2}) &\simeq \nSucc~i_{1} \nMax \nSucc~i_{2} &i \nMax \nZero &\simeq i\\
  i_{1}\nMax i_{2}&\simeq i_{2} \nMax i_{1} &i \nMax \nSucc~i &\simeq \nSucc~i &i \nMax i &\simeq i
\end{align*}

As we will see later with Proposition~\ref{prop:semanticlvl}, the definition of $\simeq$ ensures us that two levels are convertible exactly when they are arithmetically equivalent, allowing us for instance to exchange $\nSucc~i \nMax i \nMax \nZero$ and $\nSucc~i$.

\subsection{Metatheory of $\mathbb{UPP}$} We now look at the metatheory of $\mathbb{UPP}$.

\begin{prop}
The theory $\mathbb{UPP}$ is well-formed.
\end{prop}
\begin{proof}
Can be easily verified manually, or with the help of \textsc{Lambdapi}~\cite{lambdapi}.
\end{proof}

In the following, we consider $\mathbb{UPP}$ extended with an arbitrary local signature $\Phi$ well-formed in $\mathbb{UPP}$. Therefore, the following  $\delta$ rules are the ones of $\Phi$.

\begin{prop}\label{prop:confluence}
The rewrite system $\mathcal{R}_{\mathbb{UPP}}$ is confluent together with the rules $\beta$, $\beta_{c}$ and $\delta$.
\end{prop}
\begin{proof}
Follows from the fact that the rewrite rules define an orthogonal combinatory rewrite system \cite{CRS}.
\end{proof}

The following basic property is similar to \cite[Lemma 4.1.6]{UPTS} and shows that $\red$ and $\simeq$ interact well.

\begin{prop}\label{prop:sim}
  The relation $\simeq$ is a simulation with respect to $\red$. Diagrammatically,
\[\begin{tikzcd}
	t & u \\
	{\exists t'} & {u'}
	\arrow[from=1-2, to=2-2]
	\arrow[dashed, from=1-1, to=2-1]
	\arrow["\simeq"{description}, draw=none, from=2-1, to=2-2]
	\arrow["\simeq"{description}, draw=none, from=1-1, to=1-2]
\end{tikzcd}\]
\end{prop}
\begin{proof}
  By induction on the rewrite context of $u \red u'$. The induction steps are easy, we only show the base cases.

\begin{enumerate}
  \item If $u \red u'$ with a $\delta$ rule, then $u $ is a non-confined constant, so $u = t$ and $t \red u'$.
  \item If $u = \Tm_{l_{1}}~\U_{l_{2}} \red \Ty_{l_{2}} =u'$, then we have $t = \Tm_{l'_{1}}~\U_{l'_{2}}$ with $l_{k} \simeq l'_{k}$ for $k=1,2$. Therefore, $t \red \Ty_{l_{2}'} \simeq \Ty_{l_{2}}=u'$.
  \item If $u = \nApp_{l_{1},l_{2}}~A_{1}~B_{1}~(\nLam_{l_{3}, l_{4}}~A_{2}~B_{2}~v)~w \red v~w = u'$, then  $t = \nApp_{l'_{1},l'_{2}}~A'_{1}~B'_{1}~(\nLam_{l'_{3}, l'_{4}}~A'_{2}~B'_{2}~v')~w'$ with $v \simeq v'$, $w \simeq w'$ and other relations that we will not need. Therefore, $t \red v'~w' \simeq v~w = u'$.
  \item If $u = (x. v) w \red v[x\mapsto w]$, then $t = (x.v')w'$ with $v \simeq v'$ and $w \simeq w'$. Therefore, $t \red v'[x\mapsto w']$ and $v'[x\mapsto w'] \simeq v[x\mapsto w]$ follows by stability under substitution.
  \item If $u = (i. v) l \red v[i\mapsto l]$, then $t = (i.v')l'$ with $v \simeq v'$ and $l \simeq l'$. Therefore, $t \red v'[i\mapsto l']$ and $v'[i\mapsto l'] \simeq v[i\mapsto l]$ follows by stability under substitution. \qedhere
        \end{enumerate}
      \end{proof}
\begin{rem}\label{rem:sim}
  We remark that the use of confinement was essential in the previous proof. Indeed, had $\nMax, \nSucc, \nZero$ been declared as regular constants, their associated equations would also have to be declared with regular variables, and we would have  $x \simeq x \nMax x$.  Then, this equation could be used in cases in which $x$ is instantiated by redexes: for instance, we~could~have\[
    \Tm_{\nZero}~\U_{\nZero} \simeq \Tm_{\nZero}~\U_{\nZero} \nMax \Tm_{\nZero}~\U_{\nZero} \red \Ty_{\nZero} \nMax \Tm_{\nZero}~\U_{\nZero}
  \]But then  $\Ty_{\nZero}$ is the only reduct of $\Tm_{\nZero}~\U_{\nZero}$, yet $\Ty_{\nZero} \simeq \Ty_{\nZero} \sqcup \Tm_{\nZero}~\U_{\nZero}$ does not hold, showing that Proposition \ref{prop:sim} fails in this setting. One could argue that this is not a problem because the term $\Tm_{\nZero}~\U_{\nZero} \nMax \Tm_{\nZero}~\U_{\nZero}$ is ill-typed, however in Proposition~\ref{prop:sim} we do not ask terms to be well-typed. This is because this property is the main lemma for showing Church-Rosser, which in turn is needed for proving subject reduction. So without having subject reduction at this point, we cannot yet rely on the fact that terms are well-typed.%

\end{rem}

      \newcommand\ared{\blacktriangleright}
      \newcommand\aredd{{\blacktriangleright}^{*}}
      \newcommand\aired{\blacktriangleleft}
      \newcommand\airedd{{}^{*}{\blacktriangleleft}}

  We will need the following simple property about abstract rewriting. Recall that an abstract equational rewrite system $(\ared,\sim)$ is given by a binary relation $\ared~ \subseteq X^{2}$ and an equivalence relation $\sim~ \subseteq X^{2}$. Then $(\ared,\sim)$ is said to be \textit{Church-Rosser modulo} if $x~ (\ared \cup \aired \cup \sim)^{*}~ y$ implies $x~ \aredd  \sim  \airedd ~y$, where we write juxtaposition for composition of relations and $\aired$ for the inverse of $\ared$. %

  \newcommand\textsfup[1]{\textsf{\textup{#1}}}

  \begin{prop}\label{prop:abstract-cr}
    Let $(\ared, \sim)$ be an abstract equational rewrite system. If $\ared$ is confluent and  $\sim$ is a simulation for $\ared$, then $(\ared, \sim)$ is Church-Rosser modulo.
  \end{prop}
  \begin{proof}
    If $x ~(\ared \cup \aired \cup \sim)^{*}~ y$, then we have $x~(\ared \cup \aired \cup \sim)^{n}~y$ for some $n$. We prove the result by induction on $n$, the base case being trivial. For the inductive step, we have \[
      x~ (\ared \cup \aired \cup \sim)^{n} ~z~ (\ared \cup \aired \cup \sim) ~y
    \]  for some $z$. First note that by i.h. we have $x~\aredd \sim \airedd~z$. We now have three possibilities:
    \begin{enumerate}
      \item $z\ared y$ : We have $x ~\aredd \sim \airedd \ared~ y$, and so by confluence we have $x ~\aredd \sim \aredd~\airedd~ y $. Using the fact that $\sim$ is a simulation with respect to $\ared$, we then get $x ~\aredd \sim\airedd~ y$.
      \item $z \aired y$ : We have $x ~\aredd \sim \airedd \aired~ y$, and thus $x ~\aredd \sim \airedd~ y $.
      \item $z \sim y$ : We have $x ~\aredd \sim \airedd \sim y$, thus using the fact that $\sim$ is a simulation with respect to $\ared$, we get $x ~\aredd \sim\airedd~ y$. \qedhere
     \end{enumerate}

  \end{proof}

  \begin{cor}[Church-Rosser modulo]
    If $t \equiv u$ then $t \redd t' \simeq u' \iredd u$.
  \end{cor}
  \begin{proof}
    Direct consequence of Propositions \ref{prop:confluence} and \ref{prop:sim} and \ref{prop:abstract-cr}.
  \end{proof}

  The previous corollary has two important consequences. First, in order to check $t \equiv u$ we do not need to employ matching modulo $\simeq$, but instead only regular syntactic matching. This is important because it means that in order to decide $\equiv$ we do not need to design a specific matching algorithm for $\simeq$, but only to decide $\simeq$ (which is indeed decidable by Theorem~\ref{cor:lvl-decidable}) and to show $\red$ to be strongly normalizing for well-typed terms (a property we conjecture to be true). Second, using Church-Rosser modulo we can prove subject reduction, a property that will be essential to show soundness of elaboration (in particular, it is used in Theorem~\ref{thm:elab-sound}).

  \begin{prop}[Subject Reduction]
    If $\Gamma \vdash t : A$ and $t \red t'$ then $\Gamma \vdash t' : A$.
  \end{prop}
  \begin{proof}
    From Church-Rosser modulo we get the injectivity of the framework's dependent function type, so from Proposition \ref{prop:basic-meta} we conclude subject reduction of $\delta$, $\beta$ and $\beta_{c}$.

    By Proposition~\ref{prop:basic-meta} once again, to show subject reduction for $\mathcal{R}_{\mathbb{UPP}}$ it suffices to prove that all of its rewrite rules preserve typing, which easily follows from inversion of typing and Church-Rosser modulo. %
  \end{proof}

  \begin{rem}
    The proof of subject reduction relies on Church-Rosser modulo, which in turn is shown using Proposition~\ref{prop:sim}. As discussed in Remark~\ref{rem:sim}, the proof of this proposition crucially relies on the use of confinement, so it is less clear how to derive Church-Rosser modulo in a setting without it. Because $\mathcal{R}_{\mathbb{UPP}}$ is left-linear and there are no critical pairs between $\mathcal{R}_{\mathbb{UPP}}$ and $\mathcal{E}_{\mathbb{UPP}}$, nor between $\mathcal{R}_{\mathbb{UPP}}$ and itself, one possibility would be to apply \cite[Theorem 5.11]{mayr1998higher} to show  the restriction of $\equiv$ to strongly normalizing (s.n.) terms to be Church-Rosser modulo. However, to show  Church-Rosser modulo for the well-typed restriction of $\equiv$ and then subject reduction we would then need to show $\red$ to be s.n. for well-typed terms. The use of confinement is thus an interesting alternative to this option, as it allows us to establish the main correctness property of elaboration (Theorem~\ref{thm:elab-sound}) without relying on strong normalization.
  \end{rem}

\section{Universe-polymorphic elaboration}
\label{sec:alg}

We are now ready to define the (partial) transformation of a local signature $ \Phi $ to the theory $ \mathbb{UPP} $. As hinted at the end of Section~\ref{sec:firstlook}, our translation works by incrementally trying to elaborate each entry of $\Phi$ into a universe-polymorphic one in $\mathbb{UPP}$, such that at the end we get a local signature  $\Phi'$ well-formed in $\mathbb{UPP}$, if no errors are produced in the process.

In order to explain it, we now suppose that we have already translated a local signature $\Phi$ to a local signature $\Phi'$ well-formed in $\mathbb{UPP}$, and  try to translate a new entry either of the form $ c : A$ or $c : A := t$. To illustrate the steps of the process, we will make use of a running example. Note that, following our previously established convention, we keep some arguments implicit in order to improve readability.
\begin{exa}
  The last example of Section \ref{sec:firstlook} is the following local signature, which is well-formed in $\mathbb{I}$.
  \begin{align*}
&\textsf{id} : \Tm~(\nPi A : \U_{\Omega}. A \nTo A)\\
    &\textsf{id} := \nLam A. \nLam x. x\\
    &\textsf{id}' : \Tm~(\nPi A : \U_{\Omega}. A \nTo A)\\
  &\textsf{id}' := \textsf{id}_{\nApp}(\nPi A : \U_{\Omega}. A \nTo A)_{\nApp} \textsf{id}
\end{align*}
Let us call $\Phi$ the local signature containing only the first entry and suppose that it has already been translated, giving the following local signature $ \Phi' $.
\begin{align*}
    &\textsf{id} : (i : \nLvl) \to \Tm~(\nPi A : \U_{i}. A \nTo A)\\
  &\textsf{id} := i.\nLam A. \nLam x. x
\end{align*}
Therefore, as a running example, we will translate step by step the second entry $ \textsf{id}' $.
\end{exa}

\subsection{Schematic terms}

The source syntax of our elaborator will be a subset of the one defined by~$\mathbb{UPP}$ extended with $\Phi'$. The only allowed levels will be confined variables $i$, appearing as arguments of constant applications, and there should be no occurrences of confined abstractions $i.t$ or  confined function types $(i:A)\to B$ --- we henceforth call such terms \textit{schematic terms}.  Therefore, in order to translate proofs from $\mathbb{I}$ the first step is defining a translation to the syntax of schematic terms, given by Figure~\ref{fig:erase}. Note that the translation is not defined for terms of the form $t~l$ or $i.t$ or $(i: A) \to B$ because these are not used in the theory $\mathbb{I}$. Moreover, because the first step is erasing all sort information, this actually defines a translation starting from any theory defined by instantiating Figure \ref{fig:dk-pts} with a specification. Therefore, it can also be applied to theories using much more complex universe hierarchies, such as those of the proof assistants \textsc{Matita} or \textsc{Coq}.

\begin{figure}
\begin{minipage}{0.45\linewidth}
\begin{align*}
  &|x| := x\\
  &|t~u| := |t| |u|\\
  &|x.t| := x.|t|\\
  &|(x:A) \to B| := (x:|A|) \to |B|\\
  &|\Type| := \Type\\
  &|\Kind| := \Kind\\
  &|c| := c~i_{1}..i_{k} \quad\text{if } c : (j_{1}..j_{k} : \nLvl) \to A \in \Phi'
\end{align*}
\end{minipage}
\begin{minipage}{0.45\linewidth}
\begin{align*}
  &|\Ty_{s}| := \Ty~i\\
  &|\Tm_{s}| := \Tm~i\\
  &|\U_{s}| := \U~i\\
  &|\nPi_{s,s'}| := \nPi~i~i'\\
  &|\nLam_{s,s'}| := \nLam~i~i'\\
  &|\nApp_{s,s'}| := \nApp~i~i'
\end{align*}
\end{minipage}
\caption{Translation to the syntax of schematic terms\\ (each inserted level variable is assumed to be different from the previous ones)}
\label{fig:erase}
\end{figure}

Let us explain the translation of Figure \ref{fig:erase}. First, all of the constants in the definition of $\mathbb{I}$ are mapped to their correspondents in $\mathbb{UPP}$, except that they are also applied to fresh level variables which are to be solved during elaboration. Then, constants from $\Phi$ are translated as they are, except that they are also given fresh level variables in order to fill in the number of level arguments that they expect. Note that this  is necessary because the translation of the entries preceding the current one introduced new dependencies on level variables which were not present in $\mathbb{I}$.
\begin{exa}
  When applying $|-|$ to $\textsf{id}'$ we get the following entry.
  \begin{align*}
    &\textsf{id}' : \Tm~(\nPi A : \U_{i_{1}}. A \nTo A)\\
  &\textsf{id}' := (\textsf{id}~i_{2})_{\nApp}(\nPi A : \U_{i_{3}}. A \nTo A)_{\nApp} (\textsf{id}~i_{4})
  \end{align*}
  Note that the occurrences of $\U_{\Omega}$ have been replaced by $\U_{i}$ (which is just a notation for $\U~i$) for some fresh $i$. Moreover, because the type of $\textsf{id}$ in $\Phi'$ is $(i : \nLvl) \to \Tm~(\nPi A : \U_{i}. A \nTo A)$, each occurrence of $\textsf{id}$ is replaced by $\textsf{id}~i$ for some fresh $i$. Finally, note that because of implicit arguments we are hiding many new variables that were also inserted. For instance, the subterm $\nPi A : \U_{i_{1}}. A \nTo A$ is an implicit notation for $\nPi_{i_{5},i_{6}} A : \U_{i_{1}}. A \hspace{0.2em}\ctbmath{\leadsto}_{i_{7},i_{8}}\hspace{0.1em} A$. However, the term $(\textsf{id}~i_{2})_{\nApp}(\nPi A : \U_{i_{3}}. A \nTo A)_{\nApp} (\textsf{id}~i_{4})$ without implicit arguments would not even fit into a line, so for readability reasons we will not write it fully.
\end{exa}

\subsection{Computing constraints}

Our next step is then to define a type system for computing constraints. This is done by adapting the seminal work of Harper and Pollack~\cite{typechecking-with-universes}, with the difference that our system will be \textit{bidirectional}.

In regular bidirectional type systems, the typing judgment $\Gamma \vdash t : A$ is split into modes infer $\Gamma \vdash t \Rightarrow A$ and check $\Gamma \vdash t \Leftarrow A$. In mode infer we start with $\Gamma, t$ and we should find a type $A$ for $t$ in $\Gamma$, whereas in mode check we are given $\Gamma, t, A$ and we should check that $t$ indeed has type $A$ in $\Gamma$. Crucial in bidirectional typing is the proper bookkeeping of pre-conditions and post-conditions, which are summarized in the following table --- there, we mark inputs with $-$ and outputs with $+$.

\begin{center}
  \vspace{0.5em}
\begin{tabular}{lll}
Judgment & Pre-condition & Post-condition \\ \hline
 $\Gamma^{-} \vdash t^{-} \Rightarrow A^{+}$         & $\Gamma \vdash$     & $\Gamma \vdash t : A$      \\
  $\Gamma^{-} \vdash t^{-} \Leftarrow A^{-}$         & $\Gamma \vdash A:s$     & $\Gamma \vdash t : A$
\end{tabular}
\vspace{0.5em}
\end{center}

In order to refine these judgments with constraints, let us start with some preliminary definitions. A \textit{(level unification) problem} $\mathcal{C}$ is a set containing \textit{(level) constraints} of the form $l \peq l'$, referred to also as \textit{equations}. In the following definitions, let $\theta$ be a \textit{level substitution}, that is, one whose domain only contains confined variables. We write $\theta \vDash \mathcal{C}$ when  $l[\theta]\simeq l'[\theta]$ for all $l \peq l' \in \mathcal{C}$, in which case $\theta$ is called a \textit{unifier} (or \textit{solution}) for $\mathcal{C}$. Let $\Xi_{\theta}$ be a context containing $ j : \nLvl$ for each level variable appearing in $i[\theta]$, for each level variable $i$ inserted in the schematic terms. We then also write  $\Gamma \vdash_{\mathcal{C}} $ when $ \theta \vDash \mathcal{C}$ implies $ \Xi_{\theta},\Gamma[\theta]\vdash$ for all $\theta$, and $\Gamma \vdash_{\mathcal{C}} t : A$ when $ \theta \vDash \mathcal{C}$ implies $ \Xi_{\theta},\Gamma[\theta]\vdash t[\theta] : A[\theta]$ for all $\theta$.

Intuitively, just like the context $\Gamma$ in $\Gamma \vdash t :A$  allows us to state a typing judgment $t:A$ with typing hypotheses of the form $x : B \in \Gamma$, the set of constraints $\mathcal{C}$ in $\Gamma\vdash_{\mathcal{C}}t:A$ refines this with equational hypotheses of the form $l \simeq l'$. With this in mind, we can now give the new typing judgments in the following table. Compared with the previous table,  we now start with a set of constraints  $\mathcal{D}$, which guarantees that the pre-condition holds, and in the process we must also find constraints $\mathcal{C}$ ensuring that the post-condition holds.

\begin{center}
  \vspace{0.5em}
\begin{tabular}{lll}
Judgment & Pre-condition & Post-condition \\ \hline
 $\Gamma^{-} \uparrow \mathcal{D}^{-} \vdash t^{-} \Rightarrow A^{+}\downarrow \mathcal{C}^{+}$         & $\Gamma \vdash_{\mathcal{D}}$     & $\Gamma \vdash_{\mathcal{C}\cup\mathcal{D}} t : A$      \\
  $\Gamma^{-}\uparrow \mathcal{D}^{-} \vdash t^{-} \Leftarrow A^{-}\downarrow \mathcal{C}^{+}$         & $\Gamma \vdash_{\mathcal{D}} A:s$     & $\Gamma \vdash_{\mathcal{C}\cup\mathcal{D}} t : A$
\end{tabular}
\vspace{0.5em}
\end{center}

We now come to the definition of the bidirectional type system for computing constraints, given in Figure~\ref{fig:bidi} --- we mark inputs with $-$ and outputs with $+$. Note that it also relies on two new conversion judgments $A \equiv B \downarrow \mathcal{C}$ and $A \equiv_{\text{whnf}} B \downarrow \mathcal{C}$ used to compute a set of constraints needed for the conversion to hold. In the given rules, we write $t \red^{\text{wh}} u$ for the reduction of $t$ to a weak-head normal form (whnf) $u$, meaning that $t \redd{} u$ and, if $u \redd{}u'~a_1~\dots~a_k$ with $k$ being possibly $0$ and each of the $a_1,\dots,a_k$ being a regular or confined term, then $u'$ matches no rewrite rule left-hand side (including $\beta$, $\beta_\textsf{c}$ and $\delta$).

\begin{figure}\raggedright
  \fbox{$A^{-} \equiv B^{-} \downarrow \mathcal{C}^{+}$}
  \begin{mathpar}
    \inferrule
    { A \red^\text{wh} A'\\ B \red^\text{wh} B'\\ A' \equiv_{\text{whnf}} B' \downarrow \mathcal{C}}
    {A \equiv B \downarrow \mathcal{C}}
    \vspace{1em}
  \end{mathpar}
  \fbox{$A^{-} \equiv_{\text{whnf}} B^{-} \downarrow \mathcal{C}^{+}$}
  \begin{mathpar}
    \inferrule
    { t = x, c, s  }
    { t \equiv_\text{whnf} t \downarrow \emptyset }
    \and
    \inferrule
    { t \equiv t' \downarrow \mathcal{C} }
    { x.t \equiv_\text{whnf} x. t'\downarrow \mathcal{C}}
    \and
    \inferrule
    { A \equiv A' \downarrow \mathcal{C}_1 \\  B \equiv B' \downarrow \mathcal{C}_2 }
    { (x : A) \to B \equiv_\text{whnf} (x : A') \to B' \downarrow \mathcal{C}_1\cup \mathcal{C}_2 }
    \and
    \inferrule
    {t \equiv_\text{whnf} t' \downarrow \mathcal{C}_1 \\  u \equiv u' \downarrow \mathcal{C}_2 }
    { t~u \equiv_\text{whnf} t'~u' \downarrow \mathcal{C}_1\cup \mathcal{C}_2 }
    \and
    \inferrule
    {t \equiv_\text{whnf} t' \downarrow \mathcal{C}}
    { t~l \equiv_\text{whnf} t'~l' \downarrow \mathcal{C}\cup \{l \peq l'\}}
    \vspace{1em}
  \end{mathpar}
  \fbox{$\Gamma^{-}\uparrow \mathcal{D}^{-} \vdash t^{-} \Leftarrow A^{-}\downarrow \mathcal{C}^{+}$}
  \begin{mathpar}
    \inferrule[Switch]
  {\Gamma \uparrow \mathcal{D} \vdash t \Rightarrow A\downarrow \mathcal{C}_1\\ A \equiv B\downarrow \mathcal{C}_2}
  {\Gamma \uparrow \mathcal{D} \vdash t \Leftarrow B\downarrow \mathcal{C}_1 \cup \mathcal{C}_2}
  \and
  \inferrule[Abs]
  {C \red^\text{wh} (x : A) \to B \\\\\Gamma,x:A \uparrow \mathcal{D} \vdash t \Leftarrow B \downarrow \mathcal{C}}
  {\Gamma \uparrow \mathcal{D} \vdash x.t \Leftarrow C \downarrow \mathcal{C}}
  \vspace{1em}
\end{mathpar}
  \fbox{$\Gamma^{-} \uparrow \mathcal{D}^{-} \vdash t^{-} \Rightarrow A^{+}\downarrow \mathcal{C}^{+}$}
\begin{mathpar}
  \inferrule[Var]
  { x : A \in \Gamma }
  {\Gamma \uparrow \mathcal{D} \vdash x \Rightarrow A \downarrow \emptyset}
  \and
  \inferrule[Cons]
  { c : A \in \Sigma_{\mathbb{UPP}}, \Phi' \text{ or } c : A := t \in \Sigma_{\mathbb{UPP}}, \Phi'\\\\ A = (j_1..j_k : \nLvl) \to B}
  { \Gamma \uparrow \mathcal{D} \vdash c~i_1..i_k \Rightarrow B[\vec{j}\mapsto \vec{i}] \downarrow \emptyset }
  \\
  \inferrule[Sort]
  { }
  { \Gamma \uparrow \mathcal{D} \vdash  \Type \Rightarrow \Kind \downarrow \emptyset}
  \and
  \inferrule[Pi]
  {\Gamma \uparrow \mathcal{D} \vdash A \Leftarrow \Type  \downarrow \mathcal{C}_1 \\\\ \Gamma, x : A  \uparrow \mathcal{D} \cup\mathcal{C}_1\vdash B\Rightarrow C  \downarrow \mathcal{C}_2 \\ C \red^\text{wh} s}
  {\Gamma \uparrow \mathcal{D} \vdash (x : A) \to B \Rightarrow s \downarrow \mathcal{C}_1 \cup \mathcal{C}_2}
  \and
  \inferrule[App]
  { \Gamma \uparrow \mathcal{D} \vdash t  \Rightarrow C \downarrow \mathcal{C}_1\\ C \red^\text{wh} (x : A) \to B \\  \Gamma \uparrow \mathcal{D} \cup \mathcal{C}_1 \vdash u \Leftarrow A \downarrow \mathcal{C}_2}
  {\Gamma \uparrow \mathcal{D} \vdash t~u \Rightarrow B[x\mapsto u] \downarrow \mathcal{C}_1 \cup \mathcal{C}_2}
\end{mathpar}
\caption{Bidirectional type system for computing universe level constraints}
\label{fig:bidi}
\end{figure}

\begin{exa}
  In order to compute the constraints of the entry
  \begin{align*}
    &\textsf{id}' : \Tm~(\nPi A : \U_{i_{1}}. A \nTo A)\\
  &\textsf{id}' := (\textsf{id}~i_{2})_{\nApp}(\nPi A : \U_{i_{3}}. A \nTo A)_{\nApp} (\textsf{id}~i_{4})
  \end{align*}
  we first compute the constraints necessary for its type to be of type $\Type$:\[
    () \uparrow \emptyset \vdash \Tm~(\nPi A : \U_{i_{1}}. A \nTo A) \Leftarrow \Type \downarrow \mathcal{C}_{1}
  \]Then, once we know that the type is valid under the constraints $\mathcal{C}_{1}$, we can do the same with the term:\[
    ()\uparrow \mathcal{C}_{1} \vdash (\textsf{id}~i_{2})_{\nApp}(\nPi A : \U_{i_{3}}. A \nTo A)_{\nApp} (\textsf{id}~i_{4}) \Leftarrow \Tm~(\nPi A : \U_{i_{1}}. A \nTo A) \downarrow \mathcal{C}_2
  \]In the end we get the constraints\[
    \mathcal{C}_{1} \cup \mathcal{C}_{2} := \{\nSucc~i_{1}\peq i_{2}, i_{1}\peq i_{3}, i_{1}\peq i_{4},...\}
  \]where the hidden constraints concern level variables appearing in implicit arguments.
\end{exa}

We can show the soundness of the new type system by verifying that each rule locally preserves the invariants of the table --- this is the essence of the proof of Theorem~\ref{thm:elab-sound}. Before showing this, we first  need a lemma establishing the soundness of conversion checking.

\begin{lem}\label{lem:sound}
  Suppose that $A\equiv B \downarrow \mathcal{C}$ or $A\equiv_{\textup{whnf}} B \downarrow \mathcal{C}$. If $\theta \vDash \mathcal{C}$ then $A[\theta]\equiv B[\theta]$.
\end{lem}
\begin{proof}
By an easy mutual induction on $A\equiv B \downarrow \mathcal{C}$ and $A\equiv_{\textup{whnf}} B \downarrow \mathcal{C}$.
\end{proof}

\begin{samepage}
\begin{thm}[Soundness of type system for computing constraints]\label{thm:elab-sound} Consider the rules of Figure~\ref{fig:bidi} where $\Phi'$ is an arbitrary  local signature well-formed in $\mathbb{UPP}$.
  \begin{itemize}
    \item If $\Gamma \vdash_{\mathcal{D}} A:s$ and $\Gamma \uparrow \mathcal{D} \vdash t \Leftarrow A \downarrow \mathcal{C}$ then $ \Gamma \vdash_{\mathcal{C}\cup\mathcal{D} } t  :A$
    \item If $\Gamma \vdash_{\mathcal{D}} $ and $\Gamma \uparrow \mathcal{D} \vdash t \Rightarrow A \downarrow \mathcal{C}$ then $ \Gamma \vdash_{\mathcal{C}\cup\mathcal{D}} t  :A$
  \end{itemize}
\end{thm}\end{samepage}
\begin{proof}
  By mutual induction on the elaborator judgments.

  \begin{itemize}
    \item Case \textsc{Var} : Let $\theta \vDash \emptyset \cup \mathcal{D}$. By hypothesis we have $\Xi_{\theta},\Gamma[\theta] \vdash$, and hence $\Xi_{\theta},\Gamma[\theta]\vdash x : A[\theta]$.
    \item Case \textsc{Sort} : Similar to the previous case.
    \item Case \textsc{Cons} : Let $\theta \vDash \emptyset \cup \mathcal{D}$. By hypothesis we have $\Xi_{\theta},\Gamma[\theta] \vdash$, and hence $\Xi_{\theta},\Gamma[\theta]\vdash c :  (j_{1}..j_{k} :\nLvl) \to B$. Because $i_{n}[\theta]$ is a level with free variables among $\Xi_{\theta}$, we then also have $\Xi_{\theta},\Gamma[\theta]\vdash i_{n}[\theta]:\nLvl$, for all $n=1..k$. So we get $\Xi_{\theta},\Gamma[\theta]\vdash c~i_{1}[\theta]..i_{k}[\theta] : B[\vec{j}\mapsto\vec{i}[\theta]]$, and because $B[\vec{j} \mapsto \vec{i}[\theta]] = B[\vec{j}\mapsto\vec{i}][\theta]$ we conclude.
    \item Case \textsc{Switch} : Let $\theta \vDash \mathcal{C}_{1} \cup \mathcal{C}_{2} \cup \mathcal{D}$. By hypothesis we have $\Gamma \vdash_{\mathcal{D}} B:s$ and thus $\Gamma\vdash_{\mathcal{D}}$, therefore by i.h. we have $\Gamma \vdash_{\mathcal{C}_{1} \cup \mathcal{D}} t : A$, from which we get  $\Xi_{\theta}, \Gamma[\theta]\vdash t[\theta]: A[\theta]$. Moreover, from Lemma~\ref{lem:sound} we also get $A[\theta]\equiv B[\theta]$. Finally, from $\Gamma \vdash_{\mathcal{D}} B:s$ we have $\Xi_{\theta},\Gamma[\theta]\vdash B[\theta]:s$, and therefore we conclude $\Xi_{\theta},\Gamma[\theta]\vdash t [\theta]:B[\theta]$ by rule \textsc{Conv}.
    \item Case \textsc{Pi} : Let $\theta \vDash \mathcal{C}_{1} \cup \mathcal{C}_{2} \cup \mathcal{D}$. By hypothesis we have $\Gamma \vdash_{\mathcal{D}}$ and thus $\Gamma \vdash_{\mathcal{D}}\Type:\Kind$, therefore by i.h. we have $\Gamma \vdash_{\mathcal{C}_{1} \cup \mathcal{D}} A:\Type$, from which we get  $\Xi_{\theta}, \Gamma[\theta]\vdash A[\theta]:\Type$ and $\Gamma,x:A \vdash_{\mathcal{C}_{1} \cup \mathcal{D}}$. Therefore, by i.h. once again we also have $\Gamma,x:A \vdash_{\mathcal{C}_{1} \cup \mathcal{C}_{2} \cup  \mathcal{D}} B : C$ and thus $\Xi_{\theta},\Gamma[\theta], x : A[\theta] \vdash B[\theta]:C[\theta]$.

          We now claim that $\Xi_{\theta},\Gamma[\theta],x : A[\theta] \vdash B[\theta] : s$. Indeed, by validity we have either $C[\theta] = \Kind$ or $\Xi_{\theta},\Gamma[\theta], x : A[\theta] \vdash C[\theta] : s'$ for some $s'$. If $C[\theta] = \Kind$ then we must have $s = \Kind$, and thus the claim follows from $\Xi_{\theta},\Gamma[\theta],x : A[\theta] \vdash B[\theta]:C[\theta]$. If $\Xi_{\theta},\Gamma[\theta],x : A[\theta] \vdash C[\theta] : s'$ for some $s'$, then by subject reduction we have  $\Xi_{\theta},\Gamma[\theta],x : A[\theta] \vdash s : s'$, in which case we must have $s = \Type$ and $s' = \Kind$. Therefore,  by the conversion rule with $C[\theta]\equiv \Type$ we get $\Xi_{\theta},\Gamma[\theta],x : A[\theta] \vdash B[\theta]:\Type$.

          Finally, we conclude $\Xi_{\theta},\Gamma[\theta]\vdash (x :A[\theta]) \to  B[\theta] : s$ from  $\Xi_{\theta}, \Gamma[\theta]\vdash A[\theta]:\Type$ and $\Xi_{\theta},\Gamma[\theta],x : A[\theta] \vdash B[\theta] : s$.

    \item Case \textsc{App} : By hypothesis we have  $\Gamma\vdash_{\mathcal{D}}$, therefore by i.h. we have $\Gamma \vdash_{\mathcal{C}_{1} \cup \mathcal{D}} t : C$.

          We claim that $\Gamma \vdash_{\mathcal{C}_{1} \cup \mathcal{D}} t : (x:A)\to B$. Indeed, let $\theta \vDash \mathcal{C}_{1} \cup \mathcal{D}$, in which case we have $\Xi_{\theta},\Gamma[\theta]\vdash t[\theta]:C[\theta]$. By validity we have  $\Xi_{\theta},\Gamma[\theta]\vdash C[\theta] : s$ for some $s$, so by subject reduction and $C[\theta]\red^{*}(x:A[\theta])\to B[\theta]$ we have $\Xi_{\theta},\Gamma[\theta]\vdash (x:A[\theta]) \to B[\theta] : s$. Applying conversion with $\Xi_{\theta},\Gamma[\theta]\vdash t[\theta]:C[\theta]$, we get $\Xi_{\theta},\Gamma[\theta] \vdash t[\theta] : (x:A[\theta])\to B[\theta]$.

          Using validity and inversion of typing, we can also derive $\Gamma \vdash_{\mathcal{C}_{1} \cup \mathcal{D}} A:\Type$ from $\Gamma \vdash_{\mathcal{C}_{1} \cup \mathcal{D}} t : (x:A)\to B$, and thus we can apply the i.h. once again to get $\Gamma\vdash_{\mathcal{C}_{1}\cup\mathcal{C}_{2}\cup\mathcal{D}}u:A$.

          Now let $\theta\vDash \mathcal{C}_{1} \cup \mathcal{C}_{2} \cup \mathcal{D}$. We thus have $\Xi_{\theta},\Gamma[\theta] \vdash t[\theta] : (x:A[\theta])\to B[\theta]$ and $\Xi_{\theta},\Gamma[\theta]\vdash u[\theta]:A[\theta]$, so by the application rule we get $\Xi_{\theta},\Gamma[\theta]\vdash t[\theta]~u[\theta]:B[\theta][x\mapsto u[\theta]]$. Because $B[\theta][x\mapsto u[\theta]] = (B[x\mapsto u])[\theta]$, the result follows.

    \item Case \textsc{Abs} : By hypothesis we have $\Gamma \vdash_{\mathcal{D}} C:s $, from which we derive $\Gamma \vdash_{\mathcal{D}} (x:A) \to B : s$ using subject reduction, and then $\Gamma,x:A \vdash_{\mathcal{D}}B:s$ using inversion.

          Now let $\theta \vDash\mathcal{C}\cup\mathcal{D}$. From $\Gamma \vdash_{\mathcal{D}} (x:A) \to B : s$ we get $\Xi_{\theta},\Gamma[\theta]\vdash (x:A[\theta]) \to B[\theta]:s$, thus by inversion we have $\Xi_{\theta},\Gamma[\theta]\vdash A[\theta]:\Type$ and $\Xi_{\theta},\Gamma[\theta],x:A[\theta]\vdash  B[\theta]:s$. By the i.h. we also get $\Gamma,x:A \vdash_{\mathcal{C}\cup\mathcal{D}} t : B$, and thus $\Xi_{\theta},\Gamma[\theta],x:A[\theta]\vdash t[\theta]:B[\theta]$, and therefore we conclude $\Xi_{\theta},\Gamma[\theta]\vdash x.t[\theta] : (x:A[\theta])\to B[\theta]$ by the abstraction rule. \qedhere
  \end{itemize}
\end{proof}

\begin{rem}
  One could also wonder if the computation of constraints always terminates, either with a valid set of constraints or with an error indicating the term cannot be elaborated. By supposing strong normalization for $\mathbb{UPP}$, and by checking at each step of the rules in Figure~\ref{fig:bidi} that the constraints are consistent, one could show termination of the algorithm by using a similar technique as in \cite{typechecking-with-universes}. As we do not investigate strong normalization of $\mathbb{UPP}$ in this paper --- and doing so would further deviate us from the goals of this work ---, we leave this for future work. However, as we will see in Section~\ref{sec:matita}, when using it in practice we were able to translate many proofs without non-termination issues.
\end{rem}

\subsection{Solving the constraints} Once the constraints are computed, the next step is solving them. However, as explained in Section \ref{sec:firstlook}, we do not want a numerical assignment of level variables that satisfies the constraints, but rather a general symbolic solution which allows the term to be instantiated later at different universe levels. This thus requires unification, but because levels are not purely syntactic entities, one needs to devise a unification algorithm specific for the equational theory of universe levels. For now, let us postpone this to the next section and assume we are given a (partial) function $ \textsc{Unify} $ which computes from a set of constraints $ \mathcal{C} $ a unifier $ \theta $.

Once a unifier for the constraints is found, we can just apply it to the entry and generalize over all free level variables $\vec{i}$, and get either $c : (\vec{i} : \nLvl) \to A[\theta]$ or $c : (\vec{i}: \nLvl) \to A[\theta] := \vec{i}. t[\theta]$. However, a last optimization can be made: let us write $\vec{i}_{A[\theta]}$ for the free level variables occurring in $A[\theta]$, and $\vec{i}_{t[\theta]\setminus A[\theta]}$ for the free level variables occurring in $t[\theta]$ but not in $A[\theta]$. Then we can reduce the number of level arguments of the entry by setting all the $\vec{i}_{t[\theta]\setminus A[\theta]}$ to $\nZero$, without changing the type $A[\theta]$ in the entry: we then get $c : (\vec{i}_{A[\theta]} : \nLvl) \to A[\theta]$ or $c : (\vec{i}_{A[\theta]}: \nLvl) \to A[\theta] := \vec{i}_{A[\theta]}. t[\theta][\vec{i}_{t[\theta]\setminus A[\theta]}\mapsto \nZero]$. While this does not impact the correctness of the elaboration, this optimization is still useful because reducing the number of level arguments empirically leads to unification problems that are easier to solve in practice.

\begin{exa}
  Recall that when calculating the constraints of entry $\textsf{id}'$ we found\[
\mathcal{C}_{1} \cup \mathcal{C}_{2} := \{\nSucc~i_{1}\peq i_{2}, i_{1}\peq i_{3}, i_{1}\peq i_{4},...\}
\]The algorithm of the next section is able to compute the unifier\[
  \theta = i_{1} \mapsto i_{4}, ~i_{2} \mapsto \nSucc~i_{4},~i_{3} \mapsto i_{4}, ...
\]We can now apply the unifier to the entry and generalize over the free level variables of the type, while mapping the other ones to $\nZero$, which gives at the end
\begin{align*}
    &\textsf{id}' : (i_{4}:\nLvl)\to \Tm~(\nPi A : \U_{i_{4}}. A \nTo A)\\
  &\textsf{id}' := i_{4}.(\textsf{id}~(\nSucc~i_{4}))_{\nApp}(\nPi A : \U_{i_{4}}. A \nTo A)_{\nApp} (\textsf{id}~i_{4})
  \end{align*}
  Note that in the resulting term, the two occurrences of $\textsf{id}$ are applied to different universe levels. This illustrates the importance of the use of universe polymorphism in the translation.
\end{exa}

Let us now show the final correctness theorem for elaboration.

\begin{thm}[Correctness of elaboration] Let $\Phi'$ be a local signature well-formed in $\mathbb{UPP}$, and $A, t$ schematic terms.
  \begin{itemize}
    \item Suppose $() \uparrow \emptyset \vdash A \Leftarrow \Type\downarrow \mathcal{C}_{1}$ and $() \uparrow \mathcal{C}_{1} \vdash t \Leftarrow A \downarrow \mathcal{C}_{2}$ and $\theta=\textsc{Unify}(\mathcal{C}_{1} \cup \mathcal{C}_{2})$. Then  $\Phi', c : (\vec{i}_{A[\theta]} : \nLvl) \to A[\theta] := \vec{i}_{A[\theta]}. t[\theta][\vec{i}_{t[\theta]\setminus A[\theta]} \mapsto \nZero]$ is well-formed in~$\mathbb{UPP}$.
    \item Suppose $() \uparrow \emptyset \vdash A \Leftarrow \Type\downarrow \mathcal{C}_{1}$ and $\theta=\textsc{Unify}(\mathcal{C}_{1})$. Then  $\Phi', c : (\vec{i}_{A[\theta]} : \nLvl) \to A[\theta]$ is well-formed in $\mathbb{UPP}$.
  \end{itemize}
\end{thm}
\begin{proof} We show the first point, the proof of the second one being similar. By Theorem~\ref{thm:elab-sound}, we have $() \vdash_{\mathcal{C}_{1}\cup \mathcal{C}_{2}} t: A$, so because $\theta \vDash \mathcal{C}_{1}\cup \mathcal{C}_{2}$ we get $\Xi_{\theta}\vdash t[\theta] : A [\theta]$. Consider a substitution mapping all level variables in $\Xi_{\theta}$ not in $A[\theta]$ to $\nZero$. By the substitution property, we get $\vec{i}_{A[\theta]} : \nLvl \vdash t[\theta][\vec{i}_{t[\theta]\setminus A[\theta]}\mapsto \nZero] : A[\theta]$, and then by abstracting $\vec{i}_{A[\theta]}$ we get\[
  () \vdash \vec{i}_{A[\theta]}. t[\theta][\vec{i}_{t[\theta]\setminus A[\theta]}\mapsto \nZero] : (\vec{i}_{A[\theta]}:\nLvl) \to A[\theta]
  \]Therefore, $\Phi', c : (\vec{i}_{A[\theta]}:\nLvl) \to A[\theta] := \vec{i}_{A[\theta]}. t[\theta][\vec{i}_{t[\theta]\setminus A[\theta]}\mapsto\nZero]$ is~well-formed~in~$\mathbb{UPP}$.
\end{proof}

\section{Solving universe level unification problems}
\label{sec:solving}

Our elaborator relies on an unspecified algorithm for universe level unification, which we now present. Before going any further, let us recall that a unifier $\theta$ for a problem $\mathcal{C}$ is said to be a \textit{most general unifier} (abbreviated as \mgu{}) when, for any other unifier $\tau$ of $\mathcal{C}$, there is a substitution $\theta'$ such that $i[\theta][\theta']\simeq i[\tau]$ for all $i$ appearing in $\mathcal{C}$. When studying unification in an equational theory, the first natural question that comes to mind is whether all solvable unification problems have a most general unifier. Our first important observation is that, in the case of the equational theory of levels used in \UPP, this property does not hold.\footnote{It follows that universe level unification is not \textit{unitary}~\cite{baader1994unification}. Whether it is \textit{finitary}, \textit{infinitary} or of \textit{type zero} is still to be determined, and left for future work.}

\begin{thm}\label{no-mgu}
Not all solvable problems of universe level unification have a most general unifier.
\end{thm}

\begin{proof}
  Consider the equation $  \nSucc~i_1 = i_2 \nMax i_3 $, which is solvable, and suppose it  had a \mgu{}~$ \theta $. Note that $ \theta_1 = i_{1}\mapsto \nZero,~ i_{2}\mapsto \nSucc~\nZero,~   i_{3}\mapsto\nZero $ is also  a unifier, thus for some~$ \tau $ we have $ i_3[\theta] [\tau] \simeq \nZero $. Therefore, there can be no occurrence of $ \nSucc $ in $ i_3[\theta] $. By taking $ \theta_2 = i_1 \mapsto \nZero, i_2 \mapsto \nZero, i_3 \mapsto \nSucc~\nZero $ we can show similarly that there can be no occurrence of $ \nSucc $ in $ i_2[\theta] $. But by taking the substitution $ \theta' = \_ \mapsto \nZero$ mapping all variables to $ \nZero $, we get $ (i_2 \nMax i_3)[\theta][\theta'] \simeq \nZero $, which cannot be equivalent to $ (\nSucc~i_1)[\theta][\theta'] $. Hence, $ \nSucc~i_1 = i_2 \nMax i_3 $ has~no~\mgu{}
\end{proof}

Therefore, one cannot expect to be able to compute a \mgu{} for all solvable problems of universe level unification. One can then also wonder if, by restricting to the fragment of problems generated by the elaborator, one can expect to recover the property that all solvable problems admit a \mgu{} The following result also answers this negatively.

\begin{thm}\label{thm:term-no-mgu}
There is a schematic term whose constraints computed by the elaborator are solvable but have no most general unifier.
\end{thm}
\begin{proof}
  Consider the following term (once again, we reuse our convention of keeping some arguments implicit).
\begin{align*}
\nLam A : \U_{i_{1}}. \nLam B : \U_{i_{2}}. \nLam R : (\nPi C : \U_{i_{3}}. C \nTo C \nTo \U_{i_{4}}). R_{\nApp} \U_{i_{6}} {}_{\nApp} \U_{i_{7}} {}_{\nApp} (A \nTo B)
\end{align*}

If we try to elaborate it in mode infer in the empty context, we get a unification problem that can be simplified to $\{\nSucc~i_{7} \peq i_{1} \nMax i_{2}\}$, after solving some easy equations.\footnote{For instance, by using the algorithm of Figure~\ref{fig:unification}.} This is because the application of $R_{\nApp} \U_{i_{6}} : \Tm~(\U_{i_{6}} \nTo \U_{i_{6}} \nTo \U_{i_{4}})$ to $\U_{i_{7}}$ and $(A \nTo B)$ requires the last two to be in the same universe level, which are respectively $\nSucc~i_{7}$ and $i_{1}\nMax i_{2}$. This equation is solvable but, by Theorem~\ref{no-mgu}, does not admit a most general unifier.
\end{proof}
\begin{rem}
In the above proof, one can alternatively verify the calculation of constraints automatically in \textsc{Agda} by typechecking the code
\begin{align*}
&\textsf{test} : (A : \textsf{Set} ~\_) \to (B : \textsf{Set} ~\_) \to (R : (C : \textsf{Set} ~\_) \to C \to C \to \textsf{Set} ~\_) \to \textsf{Set} ~\_\\
&\textsf{test} = \lambda  A ~ B~ R \to  R~ (\textsf{Set} ~\_)~ (\textsf{Set} ~\_)~ (A \to B)
\end{align*}
which returns the error
\begin{center}
  \texttt{Failed to solve the following constraints: _0 $\sqcup$ _1 = lsuc _10}
\end{center}showing that \textsc{Agda}'s elaborator also simplifies the problem to find the same constraint.
\end{rem}

In other words, some schematic terms may not admit a most general universe-polymorphic instance, even when they admit some well-typed instances. A possible strategy would be to look not for a \mgu{}, but instead for a minimal set of incomparable unifiers, as is often done in the equational unification literature~\cite{baader1994unification}. However, this would not only require to duplicate each term being translated, one for each incomparable unifier, but this strategy would also risk of growing the output size exponentially. Indeed, a term using a previous translated entry that was duplicated $n$ times would then need to be elaborated multiple times, once with each of these $n$ variants.

Therefore, we instead insist in looking only for \mgu{}s, even if by Theorem~\ref{thm:term-no-mgu} this approach can fail when the problem is solvable but does not admit a \mgu{} To do this, we proceed as follows in this section.

Our main contribution, given in Subsection~\ref{subsec:characterization}, is a complete characterization of the equations $l \peq l'$ that admit a most general unifier. More precisely, our result says exactly when such an equation (1) admits a \mgu{}, in which case we also have an explicit description of one, (2) does not admit any unifier, or (3) admits some unifier but no most general one.

Our characterization yields an algorithm for solving equations which is complete, in the sense that we can always find some \mgu{} when the equation admits one. However, because we are interested in unification problems that may contain multiple equations, in Subsection~\ref{subsec:unification} we then apply this characterization in the design of a partial algorithm using a \textit{constraint-postponing} strategy~\cite{ziliani2017comprehensible, dowek1996unification, reed2009higher}: at each step, we look for an equation which admits a \mgu{} and eliminate it, while applying the obtained substitution to the other constraints. This can then bring new equations to the fragment admitting a \mgu{}, allowing us to solve them next. This is similar to how most proof assistants handle higher-order unification problems, by trying to solve the equations that are in the \textit{pattern} fragment, in the hope of unblocking some other ones in the process.

\subsection{Properties of levels} Before presenting our main results, we first start by reviewing some important properties about universe levels that will be useful in our proofs.

\begin{nota}
We adopt new notation conventions to improve the readability of large level expressions. In the following, we write $n + l$ for the level $\nSucc^{n}~l$, $n$ for the level $\nSucc^{n}~\nZero$, and we drop the \ctb{blue} color in $\sqcup$. For instance, the level $\nSucc~\nZero \nMax \nSucc~(\nSucc~i \nMax \nZero) $ will henceforth be written $1 \sqcup 1 + (1+i \sqcup 0)$ --- note that $+$ binds tighter then $\sqcup$, and that the left argument of $+$ is always a natural number, so this expression can be parsed unambiguously.
\end{nota}

In order to be able to compare levels syntactically, it is useful to introduce a notion of canonical form. A level is said to be in \textit{canonical form}~\cite{UPTS, guillaume} when it is of the form\[
  p \sqcup n_{1} + i_{1} \sqcup ... \sqcup n_{m} + i_{m}
\]with $n_{k}\leq p$ for all $ k=1..m$, and each variable occurs only once. In this case we call $p$ the \textit{constant coefficient}, and $n_{k}$ the \textit{coefficient of $i_{k}$}. We recall the following fundamental property, which appears in~\cite{UPTS, guillaume, blanqui22fscd}, and which we reprove here for completeness reasons.

\begin{thm}\label{thm:nf}
Every level is equivalent to a canonical form, which is unique modulo associativity-commutativity. Moreover, there is a computable function mapping each level to one of its canonical forms.
\end{thm}
\begin{proof}
  Given a level $l$, we first replace each variable $i$ by $i \sqcup 0$ (which are convertible levels). Then, by repeatedly applying $1+(l \sqcup l') \simeq 1 + l\sqcup 1 + l'$, we get a level of the form $l_{1}\sqcup ... \sqcup l_{p}$, in which each $l_{k}$ is either of the form $n_{k}+i_{k}$ or $n_{k}$. Note that we can easily show $n + i \sqcup i \simeq n + i$ for all $n \in \mathbb{N}$, by induction on $n$. Using this equation, we can merge all constant coefficients, and then all coefficients of a same variable, by always taking the maximum between them. Because in the beginning we started by replacing each variable $i$ by $i \sqcup 0$, it follows that the constant coefficient of the resulting level must be greater or equal to all variable coefficients, hence it is in canonical form.

  To see that the canonical form is unique modulo associativity-commutativity, it suffices to note that if two canonical forms have different coefficients for a variable $i$, then by applying a substitution mapping $i$ to some $n$ large enough and the other variables to $0$ we get two levels which are not convertible, hence the canonical forms we started with could not have been convertible. Similarly, if the constant coefficients are different, it suffices to take the substitution mapping all variables to $0$, which then also yields non-convertible levels.
\end{proof}

\newcommand{\at}[1]{\langle #1 \rangle}

We hence get the following theorem, also in \cite{UPTS, guillaume, blanqui22fscd}.

\begin{cor}\label{cor:lvl-decidable}
The equational theory $\simeq$ is decidable for levels.
\end{cor}

\newcommand\varz{j_{\nZero}}

In view of Theorem \ref{thm:nf} and the notion of canonical form, we introduce the following notation: given a level $l$, we write $l\at{i}$ for the coefficient of $i$ in its canonical form, and set it to  $-\infty$ if $i\not\in\fv(l)$. We extend this notation to $l\at{\nZero}$, which denotes the constant coefficient of the canonical form of $l$. Note that from the definition of canonical forms, we always have $l\at{\nZero}\neq -\infty$, and $l\at{\nZero}\geq l\at{i}$ for all $i$, and $l\at{i} \neq -\infty$ only for finitely many $i$ --- and moreover, an assignment  $\mathcal{I}\cup \{\nZero\} \to \mathbb{N}\cup \{-\infty\}$ defines a valid canonical form exactly when these conditions are met.

In the following, let $\varz$ stand for either a variable $j$ or the constant $\nZero$.  Then Theorem \ref{thm:nf} says exactly that $l \simeq l'$ iff for all $\varz$ we have $l\at{\varz}=l'\at{\varz}$. This principle will be very useful when proving or disproving that two levels are equivalent.

The definition of $\simeq$ can now be justified by the following property. Given a function $\phi$ mapping each confined variable to a natural number, define the interpretation $\trans{l}_{\phi}$ of a level $ l $ by interpreting the symbols $ \nZero, \nSucc $ and $ \nMax $  as zero, successor and max, and by interpreting each variable $i$ by $\phi(i)$.

\begin{prop}\label{prop:semanticlvl}
We have $l_{1} \simeq l_{2}$ iff $ \forall \phi, \trans{l_{1}}_{\phi} = \trans{l_{2}}_{\phi}$.
\end{prop}
\begin{proof}

  Note that for each $l \simeq l' \in \mathcal{E}_{\mathbb{UPP}}$ we have $\trans{l}_{\phi}=\trans{l'}_{\phi}$ for all $\phi$, and thus the direction $\Rightarrow$ can be showed by an easy induction on $l_{1}\simeq l_{2}$.

  For the other direction, let us take the canonical forms $l_{1}'$ of $l_{1}$ and $l_{2}'$ of $l_{2}$. By the left to right implication, we have $\trans{l_{1}}_{\phi}=\trans{l_{1}'}_{\phi}$ and $\trans{l_{2}}_{\phi}=\trans{l_{2}'}_{\phi}$ for all $\phi$, hence $\trans{l_{1}'}_{\phi}=\trans{l_{2}'}_{\phi}$ for all $\phi$. By varying $\phi$ over suitable valuations we can show that $l_{1}'$ and $l_{2}'$ have the same constant coefficients, and that each variable appearing in one also appears in the other with the same coefficient. Therefore, $l_{1}'$ and $l_{2}'$ are equal modulo associativity-commutativity, and thus $l_{1} \simeq l_{2}$.
\end{proof}

In other words, $\simeq$ allows one to simplify level expressions which are semantically the same --- for instance, $1+ i \sqcup i \sqcup 0$ and $1+i$. This also shows that our definition of $\simeq$, which is also used in \cite{agda}, agrees with the one used in other works about universe levels~\cite{guillaume, UPTS, gaspard, blanqui22fscd}.

\subsection{Characterizing equations that admit a m.g.u}\label{subsec:characterization} With the preliminaries now set up, we can  move to the main contribution of this section: a characterization of the equations that admit a most general unifier, along with an explicit description of a \mgu{} in these cases. Our first step is to introduce a notion of canonical form for equations.

\begin{defi}
  An equation $ l_1 \peq  l_2 $ is said to be in canonical form if
  \begin{enumerate}
  \item Both $ l_1, l_2 $ are in canonical form.
  \item If $ i \in \fv(l_{1}) \cap \fv(l_2) $, then $ l_{1}\at{i} = l_{2}\at{i}$
  \item At least some coefficient in $l_{1}$ or $ l_{2}$ is equal to $ 0 $
  \end{enumerate}
\end{defi}

The main motivation for introducing this notion is the following result, stating that in our analysis it suffices to consider only equations in canonical form.

\begin{prop}
For all equations $l_{1}\peq l_{2}$, there is an equation  $l_{1}' \peq l_{2}'$ in canonical form, such that for all $\theta$, $l_{1}[\theta]\simeq l_{2}[\theta]$ iff $l_{1}'[\theta] \simeq l_{2}'[\theta]$.
\end{prop}
\begin{proof}
  Let $l_{1}\peq l_{2}$ be any equation. We apply transformations so that properties (1)-(3) that define canonical forms are satisfied one by one, and we argue that they do not change the set of unifiers.
  \begin{enumerate}
    \item We put each level $l_{p}$ in canonical form $l'_{p}$. It is clear that this preserves the set of unifiers, as any level is convertible to its canonical form.

    \item If some variable $i$ appears in $l'_{1}$ and $l'_{2}$ with different coefficients, we remove it from the side with smaller coefficient, and we name the resulting equation $l''_{1}\peq l''_{2}$ --- this step is then repeated until condition (2) of the canonical form definition is met. By decomposing $l_{1}' \simeq l_{a} \sqcup n + i$ and $l_{2}' \simeq l_{b} \sqcup m + i$ with $n < m$ (or the symmetric), the correctness of this step follows from  $l_{1}'[\theta] \simeq l_{2}'[\theta]$ iff $l_{a}[\theta]\sqcup n + i[\theta] \simeq l_{b}[\theta] \sqcup m + i[\theta]$ iff $l_{a}[\theta] \simeq l_{b}[\theta] \sqcup m + i[\theta]$, where the last equivalence follows from the fact that\[
          \textsf{max}\{k_{a}, n + q\} = \textsf{max}\{k_{b}, m + q\} \quad \iff \quad k_{a} = \textsf{max}\{k_{b}, m + q\}
          \]for all $k_{a}, k_{b},n, m, q \in \mathbb{N}$ with $n < m$, and then by Proposition \ref{prop:semanticlvl}.

    \item Finally, if no coefficient in $l_{1}''$ or $l_{2}''$ is equal to zero, we subtract from all coefficients the value of the current minimal coefficient, and we name the resulting equation $l_{1}'''\peq l_{2}'''$. If we call this value $k$, then the correctness of this step follows from the fact that $l_{p}'' \simeq k + l'''_{p}$ for $p=1,2$, and so $l_{1}''[\theta]\simeq l_{2}''[\theta]$ iff $k + l_{1}'''[\theta]\simeq k+l_{2}'''[\theta]$ iff $l_{1}'''[\theta] \simeq l_{2}'''[\theta]$, where the last equivalence follows by applying Proposition~\ref{prop:semanticlvl}.
  \end{enumerate}
  It is clear that $l_{1}'''\peq l_{2}'''$ is in canonical form, and we have shown that each step of the transformation preserves the set of unifiers.
\end{proof}

\begin{exa}
  Consider the equation $ i\sqcup 1+(i \sqcup 1+j) \peq j \sqcup 2+i $ and let us show how it can be put in canonical form using the underlying algorithm of the above proof. First, we compute the level canonical forms of each side, yielding\[
    2 \sqcup 1+i \sqcup 2+j \peq 2 \sqcup 2+i \sqcup j
  \]As the variables $i$ and $j$ appear in the two sides with different coefficients, we then remove
  from each of the sides the occurrence with the smaller coefficient, yielding\[
    2 \sqcup  2+j \peq 2 \sqcup 2+i
  \]Finally, as the minimum among all coefficients is $ 2 $, we subtract this from all of them, giving\[
    0 \sqcup j \peq 0 \sqcup i
  \]
\end{exa}

\newcommand{\enix}[1]{[#1]}

We are now able to state the main theorem that we are going to show. In the following, if $k \in \mathbb{N}$ we write $\enix{k}$ for the set $\{1,...,k\}$, and we call an equation $l_{1}\peq l_{2}$ \textit{trivial} when $l_{1} \simeq l_{2}$. We also call $l_{2} \peq l_{1}$ the \textit{symmetric} of the equation $l_{1}\peq l_{2}$. Finally, if $\vec{i}= i_1...i_k$ is a list of level variables, we sometimes identify it with the level $i_1 \sqcup ...\sqcup i_k$.

\begin{thm}\label{thm:unif}
  A non-trivial equation has
  \begin{enumerate}
    \item[(A)] a most general unifier iff its canonical form (or its symmetric) is of the form
    \begin{enumerate}
      \item[(i)] $n \sqcup i \peq  l$ with $n < l\at{0}$, in which case $\theta = i\mapsto l$ is a \mgu{}
      \item[(ii)] $0 \sqcup \vec{i}_{0} \sqcup \vec{i}_{1} \peq  0 \sqcup\vec{i}_{0} \sqcup \vec{i}_{2}$ with $\vec{i}_{0}$, $\vec{i}_{1} $ and  $ \vec{i}_{2} $ disjoint, in which case a \mgu{} is given by
      \begin{align*}
        \theta =\quad & i_{0}^{k} \mapsto x_{k} \sqcup (\sqcup_{n \in \enix{p_1}} y_{k, n})\sqcup (\sqcup_{m \in \enix{p_{2}}} z_{k, m})  && (k \in \enix{p_{0}})\\
        &i^{n}_{1} \mapsto (\sqcup_{k \in \enix{p_{0}}} y_{k, n}) \sqcup (\sqcup_{m \in \enix{p_{2}}} v_{n, m}) && (n \in \enix{p_{1}})\\
        &i^{m}_{2} \mapsto (\sqcup_{k \in \enix{p_{0}}} z_{k, m}) \sqcup (\sqcup_{n \in \enix{p_{1}}} v_{n, m}) && (m \in \enix{p_{2}})
      \end{align*}
      where $p_{0}, p_{1}, p_{2}$ are the lengths of $\vec{i}_{0}$ and $ \vec{i}_{1}$ and $\vec{i}_{2}$ respectively, and where $\{x_{k}\}_{k \in \enix{p_{0}}}$,  $\{y_{k, n}\}_{k \in \enix{p_{0}}, n \in \enix{p_{1}}}$, $\{z_{k, m}\}_{k \in \enix{p_{0}}, m \in \enix{p_{2}}}$ and $\{v_{n, m}\}_{n \in \enix{p_{1}}, m \in \enix{p_{2}}}$ are disjoint sets of variables.
    \end{enumerate}
    \item[(B)] no unifier iff its canonical form (or its symmetric) is of the form $n \peq  l $ with $n < l\at{0}$
    \item[(C)] some unifier but no most general one iff its canonical form (or its symmetric) is not of any of the previous forms
  \end{enumerate}
\end{thm}

Before proving the result, let us consider some examples to see how it can be used.

\begin{exa}\label{exa:unif}~

  \begin{itemize}
    \item The equation $i_{0} \sqcup i_{1} \peq i_{0} \sqcup i_{2} $ has the canonical form\[
 0 \sqcup i_{0} \sqcup i_{1} \peq 0 \sqcup  i_{0} \sqcup i_{2}
          \]and therefore by point (A.ii) it admits the \mgu{}\[
          \theta = \{i_{0} \mapsto x \sqcup y \sqcup z,~i_{1}\mapsto y \sqcup v,~i_{2} \mapsto z \sqcup v \}
          \]
    \item The equation  $ i \sqcup 1 + (j \sqcup 2) \peq 1 + (2 \sqcup i \sqcup j) $ has the canonical form\[
          2\sqcup j  \peq 2 \sqcup i \sqcup j
          \]and therefore by point (C) it is solvable but admits no \mgu{}
    \item The equation $ i \sqcup 1 + (j \sqcup 1) \peq 1 + (2 \sqcup i \sqcup j) $ has the canonical form\[
          1\sqcup j \peq 2 \sqcup i \sqcup j
          \]and therefore by point (A.i) it admits the \mgu{}\[
          \theta = j \mapsto 2 \sqcup i \sqcup j
          \]
    \item The equation $ i \sqcup 1 + (j \sqcup 1) \peq 2 + (1 \sqcup i \sqcup j) $ has the canonical form\[
          0 \peq 1 \sqcup i \sqcup j
          \]and therefore by point (B) it admits no unifier.
  \end{itemize}

\end{exa}

\begin{figure}

  \includegraphics{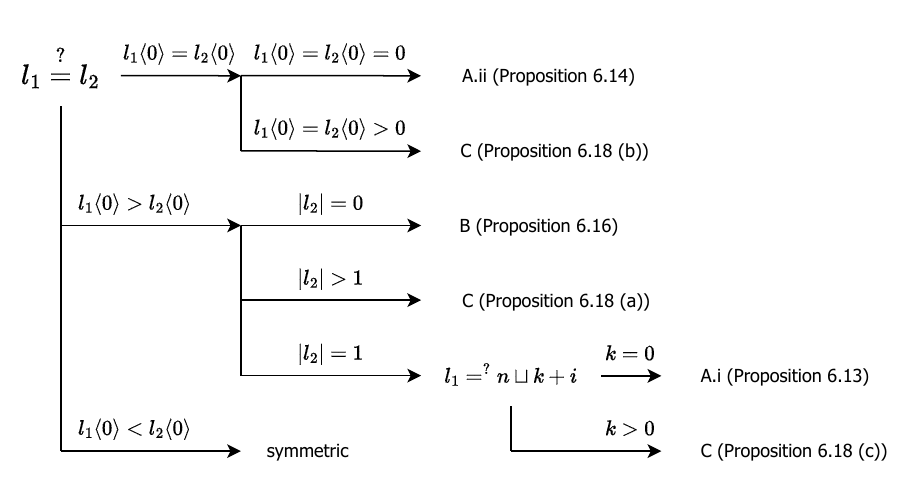}
  \caption{Structure of proof of Theorem~\ref{thm:unif}}
  \label{fig:proof}

\end{figure}

Let us now move to the proof of Theorem~\ref{thm:unif}.  Figure~\ref{fig:proof} shows its structure: we take a non-trivial equation in canonical form and consider its possible forms. Each leaf is annotated with the proposition associated with its proof, along with the case of Theorem~\ref{thm:unif} which we are in. We also write $|l_{2}|$ for the number of free variables occurring in $l_{2}$.

\subsubsection{Equations with \mgu{}s}

\begin{prop}\label{prop:mgu-1}
  The equation in canonical form $n \sqcup i \peq l $ with $n < l\at{\nZero}$ has the mgu $\tau = i \mapsto l$.
\end{prop}
\begin{proof}
  It is easy to verify that $\tau$ is a unifier. Now let $\theta$ be an arbitrary unifier and let us first show that $i[\theta]\simeq l[\theta]$. To do this we show $i[\theta]\at{\varz} = l[\theta]\at{\varz}$ for all $\varz$. Because the canonical forms of $i[\theta]$ and $n \sqcup i[\theta]$ can only differ on their constant coefficients, and because $n \sqcup i[\theta] \simeq l[\theta]$, then it follows that $i[\theta]\at{j} = l[\theta]\at{j}$ for all variables $j$. For the case $\varz = \nZero$, we have $\textsf{max}\{n, i[\theta]\at{\nZero}\} = l[\theta]\at{\nZero}$, and because $l\at{\nZero}> n$ then $l[\theta]\at{\nZero}>n$, so the only possibility is $i[\theta]\at{\nZero}=l[\theta]\at{\nZero}$.

  Now we can show that $\tau$ is more general than $\theta$: we have $j [\tau] [\theta] \simeq j [\theta] $ for all $j \in \fv(l) \cup \{i\}$. Indeed, the equation holds trivially for $i \neq j$, and for $i = j$ it follows from $l[\theta]\simeq i[\theta]$.
\end{proof}

\begin{prop}\label{prop:mgu-2}
The equation $0 \sqcup \vec{i}_{0} \sqcup \vec{i}_{1} \peq  0 \sqcup\vec{i}_{0} \sqcup \vec{i}_{2}$, with $\vec{i}_{0}$, $\vec{i}_{1} $ and $ \vec{i}_{2} $ disjoint, has the \mgu{}
\begin{align*}
\theta =\quad & i_{0}^{k} \mapsto x_{k} \sqcup (\sqcup_{n \in \enix{p_1}} y_{k, n})\sqcup (\sqcup_{m \in \enix{p_{2}}} z_{k, m})  && (k \in \enix{p_{0}})\\
                                  &i^{n}_{1} \mapsto (\sqcup_{k \in \enix{p_{0}}} y_{k, n}) \sqcup (\sqcup_{m \in \enix{p_{2}}} v_{n, m}) && (n \in \enix{p_{1}})\\
                                  &i^{m}_{2} \mapsto (\sqcup_{k \in \enix{p_{0}}} z_{k, m}) \sqcup (\sqcup_{n \in \enix{p_{1}}} v_{n, m}) && (m \in \enix{p_{2}})
                  \end{align*}
                  where $p_{0}, p_{1}, p_{2}$ are the lengths of $\vec{i}_{0}$ and $ \vec{i}_{1}$ and $\vec{i}_{2}$ respectively, and  $\{x_{k}\}_{k \in \enix{p_{0}}}$,  $\{y_{k, n}\}_{k \in \enix{p_{0}}, n \in \enix{p_{1}}}$, $\{z_{k, m}\}_{k \in \enix{p_{0}}, m \in \enix{p_{2}}}$ and $\{v_{n, m}\}_{n \in \enix{p_{1}}, m \in \enix{p_{2}}}$ are disjoint sets of variables.
\end{prop}
\begin{proof}
  It is easy to see that $\theta$ is a unifier: all introduced variables appear in both sides, with coefficient 0. Given a unifier $\tau$, define $\tau'$ by setting for each $\varz $
  \begin{align*}
    x_{k}[\tau']\at{\varz} &:= i^{k}_{0}[\tau]\at{\varz}\\
    y_{k, n}[\tau']\at{\varz} &:= \textsf{min}\{i^{k}_{0}[\tau]\at{\varz}, i^{n}_{1}[\tau]\at{\varz}\}\\
    z_{k, m}[\tau']\at{\varz} &:= \textsf{min}\{i^{k}_{0}[\tau]\at{\varz}, i^{m}_{2}[\tau]\at{\varz}\}\\
    v_{n, m}[\tau']\at{\varz} &:= \textsf{min}\{i^{n}_{1}[\tau]\at{\varz}, i^{m}_{2}[\tau]\at{\varz}\}
  \end{align*}

  Note that this assignment indeed defines for each $i'$ a canonical form $i'[\tau']$: the constant coefficient of $i'[\tau']$ is never equal to $-\infty$, it is always greater or equal than the variable coefficients of $i'[\tau']$, and $i'[\tau']\at{j}$ is different from $-\infty$ only for finitely many $j$. Let us now show that $i'[\tau] \simeq i'[\theta][\tau']$ for all $i'$ among $\vec{i}_0, \vec{i}_1, \vec{i}_2$.

  As $i^{k}_{0}[\tau]\at{\varz}$ is greater or equal than $ \textsf{min}\{i^{k}_{0}[\tau]\at{\varz}, i^{n}_{1}[\tau]\at{\varz} \}$ and $\textsf{min}\{i^{k}_{0}[\tau]\at{\varz}, i^{m}_{2}[\tau]\at{\varz}\}$ for all $n,m$, we have
  \begin{align*}
    i^{k}_{0}[\theta][\tau']\at{\varz} &= \textsf{max}(\{i^{k}_{0}[\tau]\at{\varz}\} \\
                                     &\hspace{2.2em}\cup \{\textsf{min}\{i^{k}_{0}[\tau]\at{\varz}, i^{n}_{1}[\tau]\at{\varz} \}\mid n \in \enix{p_{1}}\} \\
                                     &\hspace{2.2em}\cup \{\textsf{min}\{i^{k}_{0}[\tau]\at{\varz}, i^{m}_{2}[\tau]\at{\varz}\} \mid m \in \enix{p_{2}}\}) \\
                                     &= i^{k}_{0}[\tau]\at{\varz}
  \end{align*}
  for all $\varz$ and $k \in \enix{p_{0}}$. Therefore, we get $ i^{k}_{0}[\theta][\tau'] \simeq i^{k}_{0}[\tau] $ for all $k \in \enix{p_{0}}$.

  Because $\tau$ is a unifier, we have $0 \sqcup \vec{i}_{0}[\tau] \sqcup \vec{i}_{1}[\tau] \simeq  0 \sqcup\vec{i}_{0}[\tau] \sqcup \vec{i}_{2}[\tau]$, from which we get\[
    \textsf{max}\{\vec{i}_{0}[\tau]\at{\varz}, \vec{i}_{1}[\tau]\at{\varz} \} = \textsf{max}\{\vec{i}_{0}[\tau]\at{\varz}, \vec{i}_{2}[\tau]\at{\varz}\}
  \]for all $\varz$. Therefore, for every $n \in \enix{p_{1}}$, there is some $k \in \enix{p_{0}}$ st $ i^{k}_{0}[\tau]\at{\varz} \geq i^{n}_{1}[\tau]\at{\varz}$, or  there is some $m \in \enix{p_{2}}$ st $ i^{m}_{2}[\tau]\at{\varz}\geq i^{n}_{1}[\tau]\at{\varz}$. Hence, either we have $\textsf{min}\{i^{k}_{0}[\tau]\at{\varz}, i^{n}_{1}[\tau]\at{\varz}\} = i^{n}_{1}[\tau]\at{\varz}$ for some $k \in \enix{p_{1}}$, or we have $\textsf{min}\{i^{n}_{1}[\tau]\at{\varz}, i^{m}_{2}[\tau]\at{\varz}\} = i^{n}_{1}[\tau]\at{\varz}$  for some $m\in\enix{p_{2}}$. Therefore, we get
  \begin{align*}
    i^{n}_{1}[\theta][\tau']\at{\varz} &= \textsf{max}( \{\textsf{min}\{i^{k}_{0}[\tau]\at{\varz}, i^{n}_{1}[\tau]\at{\varz}\}\mid k \in \enix{p_{0}}\} \\
                                     &\hspace{2.2em}\cup \{\textsf{min}\{i^{n}_{1}[\tau]\at{\varz}, i^{m}_{2}[\tau]\at{\varz}\} \mid m \in \enix{p_{2}}\}) \\
                                     &= i^{n}_{1}[\tau]\at{\varz}
  \end{align*}
  for all $\varz$ and $n \in \enix{p_{1}}$.  Therefore, we get $ i^{n}_{1}[\theta][\tau'] \simeq i^{n}_{1}[\tau] $ for all $n \in \enix{p_{1}}$.

  Finally, a symmetrical reasoning shows $ i^{m}_{2}[\theta][\tau'] \simeq i^{m}_{2}[\tau] $ for all $m \in \enix{p_{2}}$.
\end{proof}

\begin{rem}
  Proposition~\ref{prop:mgu-2} shows that, when there is no occurrence of $\nSucc$ in the equation, it can be solved as an ACUI unification problem. Indeed, the \mgu{} given there is also a \mgu{} of the equation when seen as a unification problem in the theory ACUI~\cite{BAADER1988345}.
\end{rem}

\subsubsection{Unsolvable equations}

\begin{prop}\label{prop:no-unifier}
A non-trivial equation in canonical form has no solution iff it (or its symmetric) is of the form $m \peq  l$ with $m < l\at{\nZero}$.
\end{prop}
\begin{proof}
  It is clear that $m \peq  l$ with $m < l\at{\nZero}$ has no solution. For the other direction, we show that any equation not of this form has a solution.

  First note that if $l_{1}\peq l_{2}$ has variables in both sides then it is easy to build a solution. Indeed, if $i_{1} \in \fv(l_{1})$, $i_{2} \in \fv(l_{2})$ are (not necessarily distinct) variables, then $\theta = i_{1}\mapsto p - l_{1}\at{i_{1}}, i_{2}\mapsto p - l_{2}\at{i_{2}}, \_{}\mapsto 0$, where $p = \max \{l_{1}\at{\nZero},l_{2}\at{\nZero}\}$, is a solution --- note that this is also  well-defined in the case $i_{1}=i_{2}$, because for equations in canonical form this implies $l_{1}\at{i_{1}}=l_{2}\at{i_{2}}$.

  We can thus restrict our analysis to equations with one of the sides constant, of the form $ m \peq l$. Note that we can suppose that  $l$ has some variable: indeed, if $l$ is constant and $l\at{\nZero}=m$ then the equation is trivial, and if $l\at{\nZero} < m$ then its symmetric is of the form $m'=l'$ with $m' < l'\at{\nZero}$, and indeed has no solution.  Finally, it is easy to see that for $m \peq  l$ with $m \geq l\at{\nZero}$ and where $l$ has some variable $i$, we have the unifier $\theta = i \mapsto m- l\at{i}, \_{} \mapsto 0$.
\end{proof}

\subsubsection{Solvable equations not admitting a m.g.u}

The last ingredient for our proof is showing that in all other cases there is no \mgu{} In order to show this, we will use the following auxiliary lemma. In the following, we refer to a level not containing any occurrence of $\nSucc$ as~\textit{flat}.

\begin{lem}[Auxiliary lemma]\label{lemma:aux}
  Let $l_{1} \peq l_{2}$ be an equation admitting a unifier $\theta$ and a \mgu{} $\tau$.
  \begin{enumerate}
    \item If $i[\theta] = 0$ with $i \in \fv(l_1)\cup\fv(l_2)$ then $i[\tau]$ is flat.
    \item If  $j[\theta]=m>0$ with $j \in \fv(l_1)\cup\fv(l_2)$, and for $p=1..k$ we have $i_{p}[\theta] = n_{p}<m$ with $i_p \in \fv(l_1)\cup\fv(l_2)$, then if $j[\tau]$ is flat it must  contain  one variable not in any $i_{p}[\tau]$.
    \item If $i\in \fv(l_{1})$ and $j \in \fv(i[\theta])$, then for some $i' \in \fv(l_{2})$ we must have $j \in \fv(i'[\theta])$.
  \end{enumerate}
\end{lem}
\begin{proof} We show each point separately.
\begin{enumerate}
  \item Because $\tau$ is a \mgu{}, for some $\theta'$ we have $i[\tau][\theta'] \simeq i [\theta]=0$, so if $i[\tau]$ contains an occurrence of $\nSucc$ then $i[\tau][\theta']$ will also contain one, and therefore will not be convertible to $0$.
  \item Because $\tau$ is a \mgu{}, for some $\theta'$ we have $j[\tau][\theta'] \simeq j[\theta]=m$ and $i_{p}[\tau][\theta'] \simeq i_{p}[\theta]= n_{p}$ for $p = 1..k$. Now if we suppose that $j[\tau]$ is flat, then the only way to have $j[\tau][\theta'] \simeq m >0$ is if some variable  $ i'$ in $j[\tau]$ is mapped to $m$ by $\theta'$. But because $i_p[\tau][\theta']\simeq n_p < m$, it is clear that $i'$ cannot appear in any of the $i_{p}[\tau]$.
  \item Follows from the fact that, if $\theta$ is a unifier, then the variables that appear in $l_{1}[\theta]$ must also appear in $l_{2}[\theta]$. \qedhere
\end{enumerate}
\end{proof}

\begin{prop}[Equations with no mgu]\label{prop:no-mgu} The following non-trivial equations in canonical form do not admit a \mgu{}:
\begin{enumerate}
  \item[(a)] $l_{1} \peq l_{2}$ with $|l_{2}|>1$ and $l_{1}\at{\nZero} > l_{2}\at{\nZero}$.
  \item[(b)] $l_{1} \peq l_{2}$ with $l_{1}\at{\nZero}=l_{2}\at{\nZero}>0$.
  \item[(c)] $l \peq n \sqcup k+i $ with $k >0$ and $n < l\at{\nZero}$.
\end{enumerate}
\end{prop}
\begin{proof}
The structure of the proof is the same in all cases: we suppose the existence of a most general unifier $\tau$ which we use to obtain a contradiction.
\begin{enumerate}
  \item[(a)] Let $i, j$ be two different variables in $l_{2}$. By Lemma~\ref{lemma:aux} (1), the unifiers $\theta_{1}= i\mapsto l_{1}\at{\nZero}-l_{2}\at{i}, \_{}\mapsto 0$ and $\theta_{2}= j\mapsto l_{1}\at{\nZero}-l_{2}\at{j}, \_{}\mapsto 0$ show that  $i'[\tau]$ is flat for all $i'$. But then we have $l_{1}[\tau][\_{}\mapsto 0] \simeq l_{1}\at{\nZero}$ and $l_{2}[\tau][\_{}\mapsto 0] \simeq l_{2}\at{\nZero}$, and because $\tau$ is a unifier we must then have $l_{1}\at{\nZero}=l_{2}\at{\nZero}$, a contradiction with $l_{1}\at{\nZero}>l_{2}\at{\nZero}$.
  \item[(b)] First note that $\_{}\mapsto 0$ is a unifier, so by Lemma~\ref{lemma:aux} (1),  $i'[\tau]$ is flat for all $i'$. Because the equation is supposed to be in canonical form and non-trivial, some variable $i$ appears in only one side. Take such a $i$ with a minimal coefficient, which we henceforth call  $p$.

        If $p < l_{1}\at{\nZero} $, then the unifier $\theta_{1}= i \mapsto l_{1}\at{\nZero}-p, \_{}\mapsto 0$ shows, by Lemma~\ref{lemma:aux} (2),  that $i[\tau]$ contains a variable not in any $i'[\tau]$ with $i'\neq i$, a contradiction with Lemma~\ref{lemma:aux}~(3), as $i$ appears in only one side.

        Suppose now that $p = l_{1}\at{\nZero}$. Because the equation is in canonical form and the constant coefficient of each side is different from $0$, then some variable $j$ must appear with coefficient $0$. Moreover, because the minimal coefficient of a variable occurring in only one side is $p \neq 0$, it follows that $j$ must  appear in  both sides (both occurrences, of course, with coefficient $0$). Because $p = l_{1}\at{\nZero} > 0$, by Lemma~\ref{lemma:aux} (2) the unifier $\theta_{2}=i \mapsto 1, j \mapsto p+1, \_{}\mapsto 0$ shows that some variable $i' \in \fv(i[\tau])$ does not appear in any $j'[\tau]$ with $j'$ different from $j$ and $i$. Therefore, because $i'$ can only also occur in $j[\tau]$, and because the coefficient of $i$ is $p$ and the coefficient of $j$  is $0$,  by composing $\tau$ with $i' \mapsto 1, \_{}\mapsto 0$ we get $p + 1$ at the side in which $i$ occurs but $p$ at the other side, a contradiction.

  \item[(c)] Because we suppose the equation is in canonical form, some $j$ different from $i$ must occur in $l$ with coefficient $0$. By Lemma~\ref{lemma:aux} (1), the unifier $\theta_{1}=i \mapsto l\at{\nZero}-k, \_{}\mapsto 0$ shows that $j[\tau]$ is flat, and by Lemma~\ref{lemma:aux} (2) the unifier $\theta_{2}=i \mapsto l\at{\nZero}-k, j \mapsto l\at{\nZero}, \_{}\mapsto 0$ shows that some variable in $j[\tau]$ does not occur in $i[\tau]$, given that $l\at{\nZero}-k < l\at{\nZero}$. Because $i$ is the only variable that appears in the right side, this establishes a contradiction with Lemma~\ref{lemma:aux} (3).\qedhere
\end{enumerate}
\end{proof}

\subsubsection{Putting everything together}

\begin{proof}[Proof of Theorem~\ref{thm:unif}]

  We proceed as illustrated in Figure~\ref{fig:proof}. The case $l_{1}\at{\nZero}= l_{2}\at{\nZero}$ is covered by Proposition~\ref{prop:mgu-2} when $l_{1}\at{\nZero}=l_{2}\at{\nZero}=0$, and by Proposition~\ref{prop:no-mgu}~(b) when $l_{1}\at{\nZero}=l_{2}\at{\nZero}\neq 0$. In the case $l_{1}\at{\nZero}\neq l_{2}\at{\nZero}$ we suppose w.l.o.g. that $l_{1}\at{\nZero} > l_{2}\at{\nZero}$, the other case being symmetric. Then we branch on the number of variables occurring in $l_{2}$: the case of no variables is covered by Proposition~\ref{prop:no-unifier}, and the case of more than one variable is covered by Proposition~\ref{prop:no-mgu}~(a). For the case of exactly one variable, we branch on the coefficient of this only variable. If the coefficient is zero, the result follows from Proposition~\ref{prop:mgu-1}, otherwise it follows by Proposition~\ref{prop:no-mgu}~(c).
\end{proof}

\subsection{The unification algorithm}\label{subsec:unification} We can now apply Theorem~\ref{thm:unif} in the design of a partial algorithm for universe level unification. The \textit{configurations} of our algorithm are either of the form $\bot$, or $\mathcal{C};\theta$ where $\theta$ is idempotent and $\dom(\theta)$ and $\fv(\mathcal{C})$ are disjoint. Configurations are then rewritten according to the rules of Figure~\ref{fig:unification}, where $\mathcal{C}[\sigma] :=\{l_{1}[\sigma]\peq l_{2}[\sigma] \mid l_{1}\peq l_{2} \in \mathcal{C}\}$ and $\theta[\sigma]:=\{i \mapsto i[\theta][\sigma] \mid i \in \dom(\theta)\}$. We also define the free variables of a configuration $\mathcal{C};\theta$ by $\fv(\mathcal{C};\theta) := \fv(\mathcal{C}) \cup \dom(\theta) \cup \vrange(\theta)$, where $\vrange(\theta) := \cup_{i \in \dom(\theta)} \fv(i[\theta])$.

\begin{figure}[ht]
\begin{align*}
  &(\textsc{Solve}) &\{l_1 \peq  l_2\} \cup \mathcal{C} ; \theta &\leadsto \mathcal{C}[\sigma] ; \sigma\cup \theta[\sigma] &&\text{if }\sigma = \textsf{mgu}(l_1,l_2)\\
&(\textsc{Fail}) &\{l_1 \peq  l_2\} \cup \mathcal{C} ; \theta &\leadsto \bot &&\text{if $l_1$ and $l_2$ are not unifiable}
\end{align*}
\caption{Unification algorithm for universe levels}
\label{fig:unification}
\end{figure}

Theorem~\ref{thm:unif} is used in $(\textsc{Solve})$ to detect if some equation admits a \mgu{}, and in $(\textsc{Fail})$ to detect if some equation is unsolvable. We assume that for each $\sigma$ chosen in step (\textsc{Solve}) we have $\dom(\sigma)= \fv(l_1, l_2)$, which guarantees that $\fv(\mathcal{C};\theta)\subseteq \fv(\mathcal{C}';\theta')$ whenever $\mathcal{C};\theta\leadsto \mathcal{C}';\theta'$\footnote{Note that this property is not true in Robinson's unification algorithm, because the step eliminating a trivial equation $\{x\peq x\}\cup \mathcal{C}; \theta \leadsto \mathcal{C};\theta$ may decrease the set of free variables of the configuration.}. This assumption can always be satisfied by removing useless entries $i \mapsto l$ in~$\sigma $ for which $i \not \in \fv(l_1,l_2)$, and adding trivial entries $i \mapsto i'$ for $i \in \fv(l_1,l_2)\setminus\dom(\sigma)$ and $i'$~fresh, and the resulting substitution is still a \mgu{} We also suppose that the set $\vrange(\sigma)$ only contains fresh variables, which guarantees that steps preserve idempotency~of  $\theta$ and disjointness of $\dom(\theta)$ and $\fv(\mathcal{C})$. This assumption can always be satisfied by composing~$\sigma$~with a bijective renaming, as the composition of a \mgu{} with a bijective renaming is also a \mgu{}

The algorithm succeeds if it reaches a configuration of the form $\emptyset; \theta$, it fails if it reaches the configuration $\bot$ and it gets stuck if it reaches any other  configuration in which no rule applies. Moreover, the following straightforward result guarantees that the algorithm cannot run forever, so these are the only options.

\begin{prop}
The algorithm always terminates.
\end{prop}
\begin{proof}
Each step of $(\textsc{Solve})$ decreases the cardinality of $\mathcal{C}$, and a step $(\textsc{Fail})$ leads to a final state.
\end{proof}

In the following, we write $l_{1}\peq l_{2} \in\mathcal{C};\theta$ when either $l_{1}\peq l_{2}\in\mathcal{C}$ or $l_1=i$ and $l_2 = l$ for some $i\mapsto l\in\theta$. We then write $\tau\vDash\mathcal{C};\theta$ when $l_{1}[\tau]\simeq l_{2}[\tau]$ for every $l_{1}\peq l_{2}\in\mathcal{C};\theta$.  Finally, given substitutions $\tau, \tau'$ and a set of variables $X$, we write $\tau =_X \tau'$ if $i[\tau]=i[\tau']$ for all~$i\in X$.

\begin{lem}[Key lemma]\label{lem:key}
  Suppose $\mathcal{C}_{1}; \theta_{1} \leadsto \mathcal{C}_{2};\theta_{2}$. Then
  \begin{enumerate}
    \item  $\tau \vDash \mathcal{C}_{1}; \theta_{1}$ and $\dom(\tau) \subseteq \fv(\mathcal{C}_1;\theta_1)$ imply $\tau' \vDash \mathcal{C}_{2};\theta_{2}$ for some $\tau'$ with $\tau' =_{\fv(\mathcal{C}_1;\theta_1)} \tau$ and~$\dom(\tau')\subseteq\fv(\mathcal{C}_2;\theta_2)$
    \item $\tau \vDash \mathcal{C}_{2};\theta_{2}$ implies $\tau\vDash \mathcal{C}_{1};\theta_{1}$
  \end{enumerate}
\end{lem}
\begin{proof}
  The only possible case is rule (\textsc{Solve}):\[
    \{l_{1} \peq  l_{2}\} \cup \mathcal{C} ; \theta \leadsto \mathcal{C}[\sigma] ; \sigma \cup \theta[\sigma]
  \]where $\sigma$ is a \mgu{} of $l_{1}$ and $l_{2}$. We show each point separately.

  \begin{enumerate}
    \item By hypothesis we have $l_{1}[\tau] \simeq l_{2}[\tau]$, so because $\sigma$ is a \mgu{} for $l_{1}$ and $l_{2}$ it follows that for some $\theta'$ we have $i[\sigma][\theta'] \simeq i[\tau]$ for  all $i \in \fv(l_1,l_2)$. In the following, we suppose wlog that $\dom(\theta')\subseteq\vrange(\sigma)$ --- otherwise we just take the restriction of $\theta'$ to $\vrange(\sigma)$, and the equation $i[\sigma][\theta'] \simeq i[\tau]$ still holds.

    We first claim that for some $\tau'$ we have $i[\sigma][\tau'] \simeq i[\tau]$ for all $i\in\fv(\{l_1\peq l_2\}\cup\mathcal{C};\theta)$. Because  $\dom(\tau)\subseteq\fv(\{l_1 \peq l_2\}\cup\mathcal{C};\theta)$ and $\dom(\theta')\subseteq\vrange(\sigma)$ and  $\vrange(\sigma)$ only contains fresh variables, we have $\dom(\theta')\cap\dom(\tau)= \emptyset$, allowing us to define $\tau':= \tau\cup\theta'$. If $i \in \fv(l_1,l_2)$, then  $i[\sigma]$ only contains fresh variables, and because $\tau'$ and $\theta'$ agree on fresh variables, we have $i[\sigma][\tau']=i[\sigma][\theta']\simeq i[\tau]$. Finally, if $i \in \fv(\{l_1\peq l_2\}\cup\mathcal{C};\theta)\setminus\fv(l_1,l_2)$ then because $\fv(l_1,l_2)=\dom(\sigma)$ we have $i[\sigma][\tau']=i[\tau']$, and because $\tau$ and $\tau'$ agree on non-fresh variables we have $i[\tau']=i[\tau]$.

    By the above claim, for each $l'_{1} \peq  l'_{2} \in \mathcal{C};\theta$ we have $l'_{p}[\tau] \simeq l'_{p}[\sigma][\tau']$ for $p=1,2$, so $\tau \vDash \mathcal{C};\theta$ implies $\sigma[\tau'] \vDash \mathcal{C}; \theta$ and thus $\tau' \vDash \mathcal{C}[\sigma];\theta[\sigma]$. Finally, the claim also implies $i[\sigma][\tau']\simeq i[\tau']$ for all $i \in \dom(\sigma)$, given that $i[\tau] = i[\tau']$ for $i$ not fresh. Hence, we conclude $\tau' \vDash \mathcal{C}[\sigma];\sigma \cup \theta[\sigma]$ as required.

    \item By hypothesis we have $i[\tau] \simeq i[\sigma][\tau]$ for all $i \in \dom(\sigma)$, and the equation trivially holds for $i \not \in \dom(\sigma)$, given that in this case $i[\sigma]=i$. Therefore, from $\tau \vDash \mathcal{C}[\sigma]; \theta[\sigma]$ we get  $\tau \vDash\mathcal{C};\theta$. Finally, because $\sigma$ unifies $l_{1}$ and $l_{2}$, we have $l_{1}[\sigma][\tau]  \simeq l_{2}[\sigma][\tau] $, and because $l_{p}[\tau]\simeq l_{p}[\sigma][\tau]$ for $p=1,2$, we get $l_{1}[\tau]\simeq l_{2}[\tau]$. We conclude $\tau \vDash \{l_{1}\peq l_{2}\}\cup\mathcal{C};\theta$. \qedhere
  \end{enumerate}
\end{proof}

The key lemma then leads to the correctness of the unification algorithm.

\begin{thm}[Correctness of unification]
If $\mathcal{C}; \emptyset \leadsto^{*}\emptyset;\theta$ then $\theta$ is a most general unifier for $\mathcal{C}$, and if $\mathcal{C};\emptyset \leadsto^{*} \bot$ then $\mathcal{C}$ has no unifier.
\end{thm}
\begin{proof}
  Suppose that $\mathcal{C}; \emptyset \leadsto^{*} \emptyset; \theta$. Because $\theta$ is idempotent, we have $\theta \vDash \emptyset; \theta$, so by iterating Lemma~\ref{lem:key} we get $\theta \vDash \mathcal{C};\emptyset$, showing that $\theta$ is a unifier for $\mathcal{C}$. To see it is a most general one, consider any other unifier $\tau$, and let $\tau'$ be its restriction to variables occurring in $\mathcal{C}$. Then by iterating Lemma~\ref{lem:key} we get $\tau'' \vDash \emptyset;\theta$ for some  $\tau'' =_{\fv(\mathcal{C})} \tau'$, hence $i[\tau''] \simeq i[\theta][\tau'']$ for all $i \in \dom(\theta)$. But because this equation also holds trivially for $i \not \in \dom(\theta)$, we get $i[\tau'']\simeq i[\theta][\tau'']$ for all $i$. Finally, because $\tau$ and $\tau'$ and $\tau''$ all agree on $\fv(\mathcal{C})$, then we get $i[\tau] \simeq i[\theta][\tau'']$ for all $i \in \fv(\mathcal{C})$.

  Now suppose that $\mathcal{C};\emptyset \leadsto^{*} \bot$. Then we have $\mathcal{C};\emptyset \leadsto^{*} \mathcal{C}';\theta' \leadsto \bot$.   If  $\tau$ is a unifier for $\mathcal{C}$, then by iterating Lemma~\ref{lem:key} with the restriction of $\tau$ to $\fv(\mathcal{C})$, we get a unifier for $\mathcal{C}'$. But if $\mathcal{C}';\theta' \leadsto \bot$, then $\mathcal{C}'$ must contain an unsolvable equation, a contradiction.
\end{proof}

\begin{rem}
We note that the correctness proofs do not rely on any specificity of the equational theory of universe levels, and therefore the algorithm of Figure~\ref{fig:unification}  can be used with any equational theory in which one can compute a \mgu{} for two terms when it exists.
\end{rem}

Because our algorithm uses Theorem~\ref{thm:unif}, which gives a complete characterization of the equations that admit a \mgu{}, it follows that our algorithm is complete for solving equations, in the sense that it can always find a \mgu{} for an equation that admits one. We can then wonder whether if it is also complete for problems that contain more than one equation. The following example shows that this is not the case.

\begin{exa}\label{exa:mgu}
  Consider the problem $\mathcal{C} := \{1+i_{0} \peq i_{2}\sqcup 1+i_{1},~ 1+i_{0} \peq  i_{1} \sqcup 1+i_{2}\}$. We can check that, according to Theorem~\ref{thm:unif}, both equations are solvable but admit no most general unifiers, so neither the step (\textsc{Solve}) nor (\textsc{Fail}) apply. Nevertheless, by combining both equations we get $i_{2}\sqcup 1 + i_{1}  \peq i_{1}\sqcup 1 + i_{2}$, whose canonical form is $0 \sqcup i_{1}\peq 0 \sqcup i_{2}$. Therefore, $\mathcal{C}$ is equivalent to $\mathcal{C} \cup \{0 \sqcup i_{1}\peq 0 \sqcup i_{2}\}$, a problem that can be solved by our algorithm, yielding the \mgu{} $\theta =i_{1}\mapsto 0 \sqcup i_{2},~i_{0} \mapsto 0  \sqcup  i_{2} $. It  follows that $\mathcal{C}$ also admits $\theta$ as a \mgu{}, yet our algorithm does not return any \mgu{}, showing it is not complete for problems with more than one equation.
\end{exa}

Moreover, Theorem~\ref{thm:term-no-mgu} shows that, even if our algorithm were complete, it would still get stuck in problems which are solvable but admit no \mgu{} In practice, it is very unsatisfying for the unification to get stuck, as this means that the whole predicativization algorithm has to halt. Thus, in order to prevent this, in our implementation we extended the unification with heuristics that are \textit{only} applied when none of the presented rules applies. Then, whenever the heuristics are applied, the computed substitution is still a unifier, but might not be a most general one. This means that the term which generated the unification problem  can still be translated to a valid term in $\mathbb{UPP}$, but the resulting term might not be a most general universe-polymorphic instance.

\section{\textsc{Predicativize}, the implementation}
\label{sec:predicativize}

In this section we present \textsc{Predicativize}, an implementation available at \url{https://github.com/Deducteam/predicativize/} of a variant of our algorithm.

Our tool is implemented on top of \DkCheck{} \cite{saillard15phd}, a type-checker for \Dedukti{}, and thus does not rely neither on the codebase of \Agda{}, nor on the codebase of any other proof assistant. Like in the case of \Universo{} \cite{thire}, we instrument \DkCheck{}'s conversion checker in order to implement the computation of level constraints.

Because the currently available type-checkers for \Dedukti{} do not implement rewriting modulo for equational theories other than AC (associativity-commutativity), we used Genestier's encoding of the equational theory of universe levels \cite{guillaume} in order to  define the theory $\mathbb{UPP}$ in a \DkCheck{} file.

We also note that for the moment the implementation lags behind the theory in various places, in particular by still using the older unification algorithm and the previous version of $\mathbb{UPP}$ proposed in our previous work~\cite{DBLP:conf/csl/FelicissimoBB23}.

To see how everything works in practice, one can run \texttt{make running-example} which translates our running example and produces  a \Dedukti{} file \texttt{output/running_example.dk}  and an \Agda{} file \texttt{agda_output/running-example.agda}. In order to test the tool with a more realistic example, the reader can also run \texttt{make test_agda}, which translates a proof of Fermat's little theorem from the \Dedukti{} encoding of \textsc{HOL} \cite{sttfa} to $\mathbb{UPP}$.

In the following, let us give a high-level description  of some of the practical differences with the theory presented until now. %

\paragraph*{User added constraints}

As we have seen, our transformation tries to compute the most general type for a definition or declaration to be typable. However, it is not always desirable to have the most general type, as shown by the following example.

\begin{exa}\label{succ}
  Consider the local signature
\begin{align*}
 \Phi = \textsf{Nat} : \Tm~\U_\square,~ \textsf{zero} : \Tm~\textsf{Nat},~ \textsf{succ} : \Tm~\textsf{Nat} \to \Tm~\textsf{Nat}
\end{align*}
 defining the natural numbers in $\mathbb{I}$. The translation of this signature by our algorithm is
\begin{align*}
  \Phi' = ~&\textsf{Nat} : (i : \nLvl) \to \Tm~\U_i,~ \textsf{zero} : (i: \nLvl) \to \Tm~(\textsf{Nat}~i),\\
                     &\textsf{succ} : (i ~j: \nLvl) \to \Tm~(\textsf{Nat}~i) \to \Tm~(\textsf{Nat}~j)
\end{align*}
However, we normally would like to impose $ i $ to be equal to $ j $ in the type of $ \textsf{succ} $, or even to impose $ \textsf{Nat} $ not to be universe-polymorphic.
\end{exa}

In order to solve this problem, we added to \textsc{Predicativize} the possibility of adding constraints by the user, in such a way that we can for instance impose $\textsf{Nat} $ to be in the bottom universe, or $ i = j $ in the type of the successor. Adding constraints can also help the unification algorithm, which can be particularly useful for simplifying unification problems when translating definitions that do not need to be universe-polymorphic. %

\paragraph*{Rewrite rules}
\label{subsec:rewrite}

The algorithm that we presented and proved correct covers two types of entries: definitions and constants. This is enough for translating proofs written in higher-order logic or similar systems, in which every step either poses an axiom or makes a definition or proof. However, when dealing with full-fledged type theories, such as those implemented by \textsc{Coq} or \textsc{Matita}, which also feature inductive types, it is customary to use rewrite rules to encode recursion and pattern matching~\cite{assaf,gaspard,thire}. If we simply ignore these rules when performing the translation, we would run into problems as the entries that appear after may need them to typecheck.

Therefore, our implementation  extends the presented algorithm to also translate rewrite rules. In order to do this, we use \DkCheck{}'s subject reduction checker to generate constraints and proceed similarly as in the algorithm. Because this feature is still experimental, this step requires user intervention in most cases. This is done by adding new constraints over the symbols appearing in the rules, in order for their translations to be less universe-polymorphic, which helps the algorithm. This part of the translation is yet to be formally defined, and its correctness is still to be proven. Nevertheless, it has been successfully used on the translation of \textsc{Matita}'s arithmetic library to \textsc{Agda}.

\paragraph*{Agda output}
\label{subsec:agda}

\Predicativize{} produces proofs in the theory $\mathbb{UPP}$, which is a subtheory of the one implemented by the \textsc{Agda} proof assistant. In order to produce proofs that can be used by \Agda{}, we also integrated in \Predicativize{} a translator that performs a simple syntactical translation from a \Dedukti{} file in the theory $\mathbb{UPP}$  to an \Agda{} file. For instance, \texttt{make test\_agda\_with\_typecheck} translates Fermat's Little Theorem proof from HOL to \Agda{} and typechecks it.

\section{Translating Matita's arithmetic library to Agda}
\label{sec:matita}

\begin{figure}[ht]
\begin{center}
  \includegraphics[width=\textwidth]{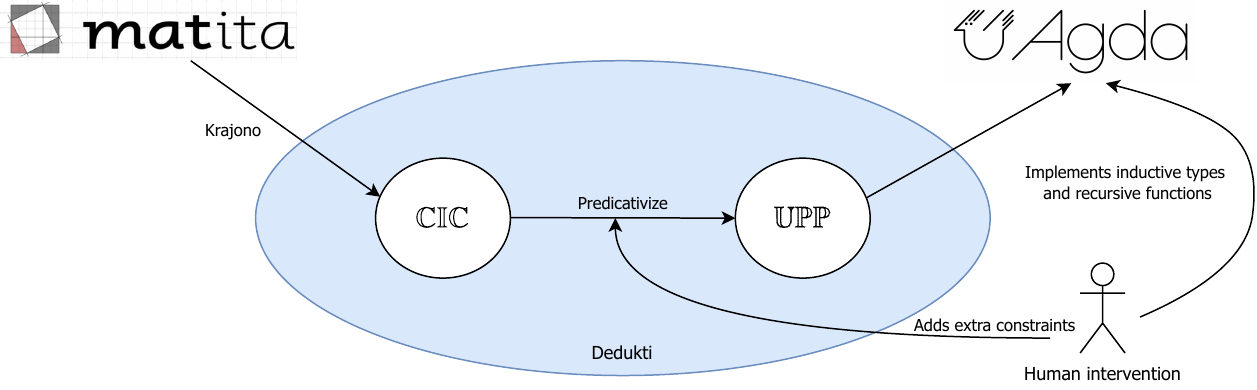}
\end{center}
\caption{Diagram representing the translation of \textsc{Matita}'s arithmetic library into \textsc{Agda}}
\label{matita}
\end{figure}

We now discuss how we used \textsc{Predicativize} to translate  \textsc{Matita}'s arithmetic library to \Agda{}. The translation is summarized in Figure \ref{matita}, where $\mathbb{CIC}$ stands for a \Dedukti{} theory defining the Calculus of Inductive Constructions, the underlying type theory of the \textsc{Matita} proof assistant.

\textsc{Matita}'s arithmetic library~\cite{matitaslib} was already available in \Dedukti{} thanks to  \textsc{Krajono}~\cite{assaf,matitadk}, a translator from \textsc{Matita} to the theory $\mathbb{CIC}$ in \Dedukti{}. Therefore, the first step of the translation was already done for us.

Then, using \textsc{Predicativize} we translated the library from $\mathbb{CIC}$ to $\mathbb{UPP}$. As the encoding of \textsc{Matita}'s recursive functions uses rewrite rules, their translation required some user intervention to add constraints over certain constants, as mentioned in the previous section. Moreover, in order to help the unification algorithm, we also added constraints for fixing the levels of many definitions which were only required to be at one universe. The list of all added constraints can be found in the file \texttt{extra_cstrs/matita.dk} in the implementation. These were obtained, for each of the concerned definitions, by looking at its (unconstrained) output and then adding equations involving some of its level variables --- similarly to how $i$ and $j$ can be equated in Example~\ref{succ}. Once this step is done, the library is known to be predicative, as it typechecks in $\mathbb{UPP}$.

We then used \textsc{Predicativize} to translate these files to \Agda{} files. However, because the rewrite rules in the \Dedukti{} files cannot be translated to \Agda{}, and given that they are needed for typechecking the proofs, the library does not typecheck directly. Therefore, to finish our translation we had to  define the inductive types and recursive functions manually in \Agda{}. To do this we first assembled the type formers and constructors, which had been translated simply as postulates, into inductive type declarations. This required us to add further constraints for some constants, for instance between $i$ and $j$ in the type of the successor (Example~\ref{succ}), in order to implement them as constructors of an inductive type.

With the inductive types defined, we could then define the recursive functions (like addition), which had been translated as postulates with no computational content. Thankfully, even if we cannot translate the rewrite rules from \Dedukti{} to \Agda{} in a way that is accepted by \Agda{}, we could still translate them as comments in the \Agda{} files. Then, instead of writing such functions from scratch, we could just adapt these comments into valid \Agda{} function declarations. We believe that this step, also needed in previous work~\cite{sttfa}, could be automated by better studying the translation between different representations of recursive functions. Nevertheless, because most of \textsc{Matita}'s arithmetic library is made of proofs, whose translation we do not need to change, automating it was not crucial in our case,  so we decided to leave this study for future work.

The result of the translation is available at
\begin{center}
  \url{https://doi.org/10.5281/zenodo.10686897}
\end{center}
and, as far as we know, contains the very first proofs in \Agda{} of Bertrand's Postulate and Fermat's Little Theorem. It also contains a variety of other interesting results such as the Binomial Law, the Chinese Remainder Theorem, and the Pigeonhole Principle. Moreover, this library typechecks with the \texttt{--safe} flag, attesting that it does not use any of \textsc{Agda}'s more exotic and unsafe features.

\begin{table}
  \centering
  \begin{tabular}{l|cccc}
  \toprule
   &  \textsc{Matita} & \textsc{Dedukti} $(\mathbb{CIC})$ & \textsc{Dedukti} $(\mathbb{UPP})$ & \textsc{Agda} \\\hline
  File size (in Kb)  &  67 & 640 &  570 &190\\
  \bottomrule
  \end{tabular}
  \caption{Comparison of (compressed) file sizes}
  \label{filesizes}
\end{table}

We conclude by discussing some statistics about the translation. The total translation time, from \textsc{Dedukti} $(\mathbb{CIC})$ to \textsc{Dedukti} $(\mathbb{UPP})$ and then to \textsc{Agda}, is about 32 minutes on a machine with an i7 processor. We also provide in Table~\ref{filesizes} a comparison of the file sizes in \textsc{Matita}, \textsc{Dedukti} (in both theories $\mathbb{CIC}$ and $\mathbb{UPP}$) and \textsc{Agda}. Here we chose to analyze their compressed sizes (using \texttt{.tar.xz}) to avoid discrepancies arising from administrative differences in the files and formats. As we see, the translation from \textsc{Matita} to \textsc{Dedukti} ($\mathbb{CIC}$) increases a lot the file sizes, which are multiplied by almost 10. This is not surprising, as the original proofs are done using tactics, that are compiled to proof terms when going to \textsc{Dedukti}. Moreover, the representation of terms in \textsc{Dedukti} is much more annotated and low-level than in commonly-used proof assistants, which also explains why some of this extra size is eliminated when going from \textsc{Dedukti} ($\mathbb{UPP}$) to \textsc{Agda}. Yet, the proofs in \textsc{Agda} are still much more low-level than their \textsc{Matita} counterparts given that they still use proof terms instead of tactics. Finally, we see that going from \textsc{Dedukti} ($\mathbb{CIC}$) to \textsc{Dedukti} ($\mathbb{UPP}$) only mildly alters the file sizes, which is not surprising since our translation does not drastically~change~the~terms.

\section{Conclusion}
\label{sec:conc}

We have proposed a transformation for sharing proofs with predicative systems. Our implementation allowed to translate many non-trivial proofs from \textsc{Matita}'s arithmetic library to \Agda{}, showing that our proposal  works well in practice.

Our solution is based on the use of universe-polymorphic elaboration. Even if elaboration algorithms are already well-studied in the literature, our proposal differs from most on the use of universe level unification, which is needed in our setting for handling universe polymorphism. Other proposals for universe-polymorphic elaboration such as \cite{typechecking-with-universes} and \cite{coq} avoid the use of universe level unification by allowing in their target languages for entries in the signature to come with associated sets of constraints, which are then verified locally at each use. This feature is however unfortunately not supported by \textsc{Agda}, the main target of our translation.

Our proposal thus required us to study the problem of universe level unification. In order to provide an algorithm for this problem, we first contributed with a complete characterization of which equations admit a \mgu{}, along with an explicit description of a \mgu{} when it exists. We then employed this characterization in the design of a unification algorithm, which is an improvement over our preliminary work~\cite{DBLP:conf/csl/FelicissimoBB23}. It is in particular able to solve all equations that admit a \mgu{}, whereas the algorithm of \cite{DBLP:conf/csl/FelicissimoBB23} was not --- for instance, it was not capable of solving the first equation of Example~\ref{exa:unif}. However some problems admitting a \mgu{} cannot be solved by our algorithm because they combine multiple equations, none of them admitting a \mgu{} (see Example~\ref{exa:mgu}). Even if our practical results show that our algorithm is already sufficiently powerful for our needs, one can wonder if a complete unification algorithm exists. We leave this interesting but difficult problem for future work.

\Agda{} also features an algorithm for solving level metavariables, but it does not seem to have been formally specified or proven correct in the literature, making it hard to provide a detailed comparison with our work. However, practical tests seem to suggest that our algorithm is an improvement. As an example, typechecking in \textsc{Agda} the entry
\begin{align*}
&\textsf{test} : (A : \textsf{Set} ~\_) \to (B : \textsf{Set} ~\_) \to (C : \textsf{Set} ~\_) \to (R : (D : \textsf{Set} ~\_) \to D \to D \to \textsf{Set} ~\_) \to \textsf{Set} ~\_\\
&\textsf{test} = \lambda  A ~ B~C~ R \to  R~ (\textsf{Set} ~\_)~ (A \to C)~ (A \to B)
\end{align*}
gives the error \texttt{Failed to solve the following constraints: _0 $\sqcup$ _1  = _0 $\sqcup$ _2}, however this constraint is solvable by our algorithm (see Example~\ref{exa:unif}). Therefore, our work could also be used to improve \textsc{Agda}'s unification algorithm.

For future work, we would also like to look  at possible ways of making \textsc{Predicativize} less dependent on user intervention. In particular, the translation of inductive types and recursive functions involves some considerable manual work. We thus expect improvements in this direction to be needed in order to translate larger proof libraries.

\section*{Acknowledgment}
  \noindent The authors would like to thank Ashish Kumar Barnawal for the very helpful discussions that led to this article, François Thiré for the help while developing \textsc{Predicativize}, Gilles Dowek for remarks about previous versions of this paper, Jesper Cockx and Vincent Moreau for discussions about universe levels and the anonymous reviewers of both CSL and LMCS for their very helpful comments and remarks.

\bibliographystyle{alphaurl}
\bibliography{ref}

\end{document}